\documentclass[lettersize,journal]{IEEEtran}
\usepackage{amsmath,amsfonts}
\usepackage{array}
\usepackage{textcomp}
\usepackage{stfloats}
\usepackage{url}
\usepackage{verbatim}
\usepackage{graphicx}
\usepackage{multirow}
\def\BibTeX{{\rm B\kern-.05em{\sc i\kern-.025em b}\kern-.08em
T\kern-.1667em\lower.7ex\hbox{E}\kern-.125emX}}
\usepackage{balance}

\usepackage[linesnumbered,ruled,vlined]{algorithm2e}
\usepackage{subcaption}
\usepackage{cite}
\usepackage{tikz}
\usepackage{tkz-tab}
\usepackage{pgfplots}
\usepackage{pgfplotstable}
\usepackage{caption}
\usepackage{amsmath}
\usepackage{longtable}
\usepackage{amssymb}
\usepackage{multicol}
\usepackage{anyfontsize}
\usepackage{xcolor}

\usepackage{amsthm}

\usetikzlibrary{shapes, snakes}
\usetikzlibrary{positioning}
\usepgfplotslibrary{statistics}
\usetikzlibrary{calc}
\usepackage{mwe}

\usepackage{amsmath,amsthm}

\newtheorem{theorem}{Theorem} 

\begin{document}
\title{Syntax Element Encryption for H.265/HEVC Using Chaotic Map-Based Coefficient Scrambling Scheme}
\author{Liang-Wei Li, Chung-Nan Lee,~\IEEEmembership{Member,~IEEE,} Kishu Gupta,~\IEEEmembership{Member,~IEEE,} Huei-Fang Yang,~\IEEEmembership{Member,~IEEE,} Ashutosh Kumar Singh,~\IEEEmembership{Senior member,~IEEE}
\thanks{Manuscript received Received 15 September 2025; revised 9 October 2025; accepted 19 October 2025. Date of publication 23 October 2025; date of current version 7 April 2026. This work was supported in part by the National Science and Technology Council, Taiwan, under Grant NSTC 114-2634-F-110-001-MBK, Grant 114-2640-E-110-004, and Grant 112-2221-E-110-047-MY3. This article was recommended by Associate Editor G. Xu. (\textit{Corresponding author: Kishu Gupta.})}
\thanks{Liang-Wei Li, Chung-Nan Lee, and Huei-Fang Yang are with the Department of Computer Science and Engineering, National Sun Yat-sen University, Kaohsiung, 80424, Taiwan. (e-mail: zaq871210@gmail.com, cnlee@mail.cse.nsysu.edu.tw, hfyang@mail.cse.nsysu.edu.tw).}
\thanks{Kishu Gupta, is with Department of Computer Science and Engineering, National Sun Yat-sen University, Kaohsiung, 804, Taiwan, and also with Department of Computer Science, the VIZJA University, 01-043 Warsaw, Poland. (e-mail: kishuguptares@gmail.com).}
\thanks{Ashutosh Kumar Singh, is with Department of Computer Science and Engineering, Indian Institute of Information Technology Bhopal, Bhopal 462003, India, and also with Department of Computer Science, the VIZJA University, 01-043 Warsaw, Poland. (e-mail: ashutosh@iiitbhopal.ac.in).}
\thanks{Digital Object Identifier 10.1109/TCSVT.2025.3625077}}

	\markboth{IEEE TRANSACTIONS ON CIRCUITS AND SYSTEMS FOR VIDEO TECHNOLOGY, VOL. 36, NO. 4, APRIL 2026}%
{Shell \MakeLowercase{\textit{Li et al.}}: A Sample Article Using IEEEtran.cls for IEEE Journals}

\makeatletter
\newcommand{\removelatexerror}{\let\@latex@error\@gobble}
\def\ps@IEEEtitlepagestyle{%
	\def\@oddfoot{\mycopyrightnotice}%
	\def\@oddhead{\hbox{}\@IEEEheaderstyle\leftmark\hfil\thepage}\relax
	\def\@evenhead{\@IEEEheaderstyle\thepage\hfil\leftmark\hbox{}}\relax
	\def\@evenfoot{}%
}

\def\mycopyrightnotice{%
	\begin{minipage}{\textwidth}
		\centering \scriptsize
		1051-8215 © 2025 IEEE. All rights reserved, including rights for text and data mining, and training of artificial intelligence and similar technologies. Personal use is permitted, but republication/redistribution requires IEEE permission. 
		\\ See https://www.ieee.org/publications/rights/index.html for more information.
		\\ This article has been published in IEEE TRANSACTIONS ON CIRCUITS AND SYSTEMS FOR VIDEO TECHNOLOGY © 2025 IEEE.
	\end{minipage}
}
\makeatother	
\maketitle

\begin{abstract}

In today's digital landscape, high-efficiency video coding (H.265/HEVC) has emerged as the most widely used video coding standard, employing selective encryption schemes to protect the privacy of video content while maintaining efficient compression performance. However, existing coefficient scrambling methods impose a significant computational load, leading to increased bit rate overhead due to encryption, longer execution times, and insufficient safety measures. To address these issues, a new coefficient scrambling scheme based on \textit{chaotic maps} is proposed. This approach leverages the pseudorandomness, ergodicity, and sensitivity to initial conditions inherent in chaotic maps to generate highly unpredictable coefficient distributions, thereby strengthening security while preserving low complexity. Unlike conventional scrambling, chaotic maps ensure minimal correlation between encrypted coefficients, enhancing resistance against statistical and differential attacks. Additionally, the scrambling conditions are specifically designed to minimize the impact on the bit rate overhead. Furthermore, when combined with syntax element encryption (SEC), which includes motion vector difference (MVD), quantized transform coefficients (QTC), and luma intraprediction mode (Luma IPM), this method effectively distorts video content. The proposed scheme operates synchronously with slices, ensuring that the decryption of video content remains intact even if some slices are lost. Additionally, a random sequence generated by AES-CTR is incorporated with the H.265 encoded stream to protect against chosen-plaintext attacks. The experimental results indicate that this scheme features high security, compliance with format standards, fast execution times, synchronous updates with slices, and resilience against common attacks, all while achieving a reduced bit rate overhead of 45.13\% with a lowered average execution time overhead of 1.91\%.
\end{abstract}


\section{Introduction}
\IEEEPARstart{W}{orldwide} rapid growth in video applications such as video-on-demand (VOD), video conferencing, online medical consultations, online courses, and video surveillance has resulted in massive volumes of streaming data. The sensitive information within these videos is at risk of leakage during transmission and storage due to the limited capabilities of streaming channels and cloud service providers. To safeguard video content from privacy erosion and copyright violations caused by data breaches, multiple solutions have been proposed \cite{c1-6316136, c2-9249233, c7-7351373, c17}. The selective encryption algorithm (SEA) encrypts parts of the H.265 encoded stream (syntax elements) to induce significant distortion while reducing computational complexity; however, it remains vulnerable to leakage of highly sensitive data due to partial content protection \cite{c1-6316136}. A recent memristive Rulkov neuron–based approach introduced an innovative chaos-driven model characterized by strong nonlinearity, a large keyspace, and efficient segment-wise encryption; however, its limited cryptographic validation, finite-precision issues, and potential data leakage constrain its robustness \cite{r4-1-11142957}. Another notable contribution provided a comprehensive review of chaos-based video encryption techniques and current research trends, yet it lacked experimental validation and relied heavily on previously reported findings \cite{r4-2-GAO2025100816}. Furthermore, critical factors such as encryption strength, computational efficiency, format compliance, bit rate overhead, and perceptual video quality must be carefully balanced to ensure secure, correct decoding and seamless playback. Inadequate consideration of these factors may lead to excessive bit rate escalation or playback distortion, undermining the overall effectiveness of video encryption schemes \cite{c19}.

In this context, this paper proposes a novel approach to furnish video content security and privacy by incorporating encryption into selective video content via chaotic map-oriented coefficient scrambling. It only encrypts five syntax elements, including \textit{motion vector difference} (MVD) information (mvd\_sign\_flag, abs\_mvd\_minus2 suffix), \textit{quantized transform coefficient} (QTC, coefficient for short) information (coeff\_sign\_flag, coeff\_abs\_level\_remaining suffix) and \textit{luma intra prediction mode} (Luma IPM) information. Furthermore, at the same time, a coefficient scrambling scheme based on chaotic maps is devised to disturb the coefficient distribution effectively. Additionally, the proposed scheme designs effective coefficient scrambling conditions ushering toward the reduced bit rate overhead rendered by coefficient scrambling. This proposed scheme uses the AES-CTR to generate a random sequence, and the encryption operation is associated with the random sequence. The generated random sequence is synchronized with the slice to ensure that the loss of some slices does not affect the decryption of the video content. In addition, this paper uses the SHA-384 hash value of the slice header to update the initial key and initial vector to resist the chosen plaintext attack.
\subsection{Key Contributions}
The main contributions of this paper are as follows:
\begin{itemize}
\item This paper presents a novel, robust H.265/HEVC video encryption scheme that aims to ensure video privacy and security. It analyzes the relationship between selective syntax elements and encryption, leveraging chaotic map-based coefficient scrambling.
\item This scheme emphasizes the impact of encryption via \textit{selective element encryption}, which involves three types of information: MVD, QTC, and Luma IPM (described in Section \ref{sec:see}), to effectively distort the video content to ensure compliance with format standards and contributes to the results, including bitrate reduction, execution time overhead reduction, and resistance to substitution attack.
\item An efficient coefficient scrambling scheme using chaotic maps, which leverages the pseudorandomness of these maps to effectively disrupt the distribution of coefficients, is presented.
\item The scheme carefully establishes the conditions for coefficient scrambling to protect video edges and minimize the bit rate overhead significantly caused by coefficient scrambling.
\item A comparative analysis, both theoretical and experimental, demonstrated significant improvements in various performance metrics compared with state-of-the-art methods.
\end{itemize}
\textit{Paper Organization:} Section \ref{sec:rel} highlights state-of-the-art approaches for Context-based Adaptive Binary Arithmetic Coding (CABAC), the Selective Encryption Algorithm (SEA), and a chaotic map. Section \ref{sec:proposed} outlines the detailed function of the proposed model. Section \ref{sec:opd} offers an operational summary of the proposed approach along with the performance metrics. Section \ref{sec:res} entails performance evaluation followed by discussion and limitations of the proposed approach in Section \ref{sec:dl}. Conclusive remarks and the future scope of the proposed work are presented in Section \ref{sec:con}. 
\section{Related Work} \label{sec:rel}
\subsection{Context-based Adaptive Binary Arithmetic Coding}
The input syntax element needs to be converted into a \textit{Bin string} through binarization, as illustrated in Fig. \ref{fig:1}. However, if the syntax element itself is represented in binary, binarization can be ignored. Binarization methods include unary code, fixed length code (FL code), truncated rice code (TR code), exponential golomb code (Exp-Golomb code), etc., followed by encoding the Bin string.
\begin{figure}[!htbp]
	\centering
	\includegraphics[width=1.0\columnwidth]{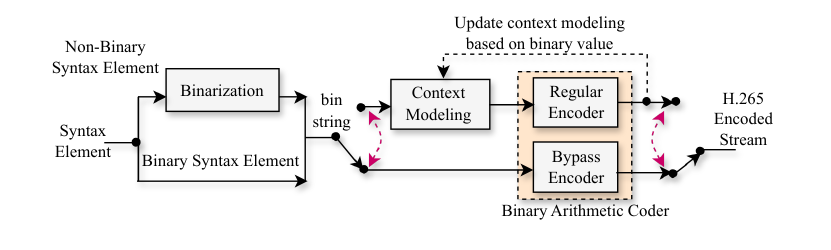}
\caption{The overall architecture of CABAC.}
	\label{fig:1}
\end{figure}

In the regular coding mode, an appropriate probabilistic model is first assigned to each bin string currently input into the context modeler based on the previously encoded bin string. Furthermore, the bin string and the newly configured probability model are input into the \textit{binary arithmetic coder} for encoding. After encoding, the context modeling needs to be updated adaptively based on the output encoding stream. In bypass coding mode, encoding is performed directly with equal probability (0 and 1 each account for half the probability).
\subsection{Selective Encryption Algorithm for H.265/HEVC}
The selective encryption algorithm (SEA) emphasizes security, execution time, format compliance, and bit rate overhead. Format compliance ensures that an encrypted video must be decoded by the decoder and that the video plays successfully, regardless of whether the decryption operation is performed. On the other hand, an encrypted video without decryption will play but appear distorted. A few specific reasons for increasing the bit rate overhead are as follows: 1) When syntax elements are encrypted in regular coding mode, modifications to the encoded data may alter the statistics of the encoded stream, and 2) directly encrypting syntax elements via some binarization method (unary code, TR code, Exp-Golomb code, etc.) can revise the length of the bin string. These prerequisites result in an increased bit rate for the encrypted video.

Lee et al. \cite{c2-9249233} proposed SEA based on the start code. However, according to Lee et al. \cite{c3-9464309}, the encrypted start code can be cracked through a ciphertext-only attack, which is not suitable for practical applications. Wallendael et al. \cite{c4-6486782} studied the encryption of many syntax elements, and the results revealed that encrypting the MVD sign and the QTC sign can cause large distortions and a constant bit rate. Peng et al. \cite{c5-8746776} proposed an H.265 adjustable SE scheme based on chroma IPM and edge coefficient scrambling. The Enc scheme encrypts multiple syntax elements and still reveals some edge information, so the Scr scheme, which adds an edge coefficient scrambling scheme but reduces the bit rate overhead, deserves further study. Sallam et al. \cite{c6-Sallam2018} used the RC6 encryption algorithm to encrypt syntax elements in different operating modes and reported that the best effect was achieved in the CFB operating mode. Farajallah et al. \cite{c7-7351373} proposed two encryption schemes based on the tile concept, locating the region of interest (ROI) in the tile. Zhang et al. \cite{c8-9345211} proposed an SE scheme based on the improved CABAC method to optimize the regular coding mode. Wen et al. \cite{c9-Wen2023} proposed three levels of H.265 SE schemes, and different levels of encrypted videos are used in different scenarios. Yang et al. \cite{c10-7489872} proposed extracting the subencoding stream from the encoded video and encrypting it to avoid the problem of reencoding for the encryption operation. Xu et al. \cite{c11-8713496} proposed an HEVC efficient exchange encryption and data hiding scheme. Boyadjis et al. \cite{c12-7370952} proposed the H.264 and H.265 luma IPM encryption methods, showing that encrypted luma IPM is sufficiently destructive to I frame content. Chen et al. \cite{c13-9903590} proposed an RSVE scheme. The generated random sequence is synchronized with the slice, and the RC4 key is updated through the slice header. The syntax elements are encrypted for MVD, QTC, and IPM information, and combined with coefficient scrambling, the two-round shift method is used to perturb the coefficients. The results show that the RSVE scheme has the characteristics of high security, resistance to common attacks, fast execution time, etc., but reducing the bit rate overhead deserves further study.

Although the RSVE scheme proposed by Chen et al. \cite{c13-9903590} is relatively robust, it lacks an improved overhead bit rate. This paper exploits the RSVE scheme as a benchmark to propose a new coefficient scrambling scheme to effectively reduce the bit rate overhead caused by coefficient scrambling. Additionally, the proposed coefficient scrambling scheme applies to the Scr scheme of Peng et al. \cite{c5-8746776} and the tunable selective encryption scheme of Sallam et al. \cite{c23-8119905}.
\subsection{Chaotic Maps for Pseudorandom Values}
Chaotic maps continually and iteratively generate pseudorandom values through a set of mathematical models. The generated values have characteristics such as ergodicity, randomness, nonperiodicity, and high sensitivity to initial conditions. Chaotic maps are commonly used to create random sequences for video encryption and can offer a large key space. Zhang et al. \cite{c14-9390925} employed a hyperchaotic Lorenz system to generate random sequences. Sallam et al. \cite{c15-Sallam2018} used random bits generated by a chaotic logistic map (CLM) to encrypt the MVD sign and the QTC sign. Liu et al. \cite{c16} exploited an integer dynamic coupling tent map to generate random sequences and supported multicore parallelization operations to increase the speed of generating random sequences. Ye et al. \cite{c17} improved the ICMIC system, combined it with CML, and proposed L-ICMIC-CML to generate random sequences. Moreover, Lu et al. \cite{c18-8977567} provided an efficient scheme to encrypt image leveraging LSS Chatoic Map and Single S-Box. 
\section{Proposed Scheme} \label{sec:proposed}
The proposed scheme for encrypting H.265/HEVC video is driven by chaotic map-based coefficient scrambling and syntax element encryption. It comprises three subparts, as illustrated in Fig. \ref{fig:2}, Fig. \ref{fig:3}, and Fig. \ref{fig:4}, collaboratively.
In the H.265 encoding process depicted in Fig. \ref{fig:2}, data generated from the intra/inter prediction, transform and quantization, and deblocking and SAO filter modules are input into the H.265 encoder. The encoder then uses CABAC to encrypt certain syntax elements. Once the encoding is completed, an encrypted H.265 stream is produced.
\begin{figure}[!ht]
	\centering
	\includegraphics[width=1.0\columnwidth]{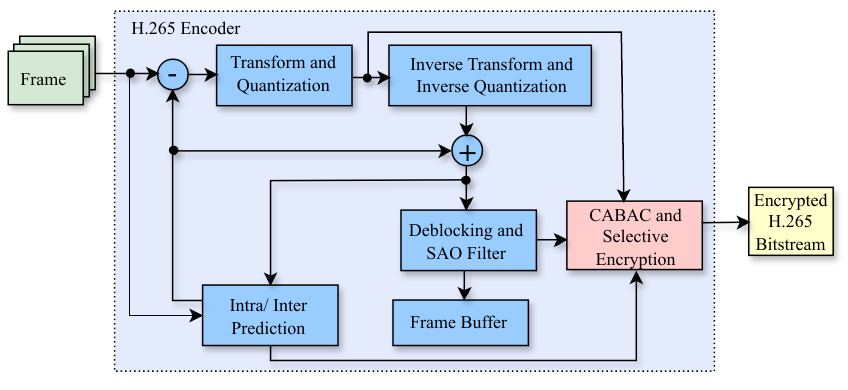}
\caption{The workflow of the proposed scheme.}
	\label{fig:2}
\end{figure}

The internal architecture of CABAC and the selective encryption module is illustrated in Fig. \ref{fig:3}. Any data input to CABAC is referred to as a syntax element. The encryption process for certain syntax elements begins after they have been converted into a binary string through binary conversion.
\begin{figure}[!ht]
	\centering
	\includegraphics[width=1.0\columnwidth]{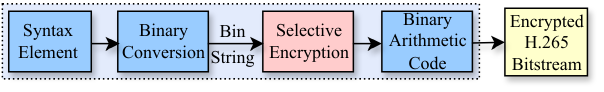}
\caption{CABAC and SEE module.}
	\label{fig:3}
\end{figure}

Furthermore, the encryption process of the proposed scheme is illustrated in Fig. \ref{fig:4}. The original slice is divided into a slice header and slice data. First, the initial key and initial vector are updated via the slice header. Then, AES-CTR is used to create a random sequence for coefficient scrambling and encryption of syntax elements. After that, coefficient scrambling is performed for the QTC, and important syntax elements in the slice data are selected for encryption. Finally, the encrypted slice data and slice header are combined to form an encrypted slice.
\begin{figure}[!ht]
	\centering
	\includegraphics[width=0.99\columnwidth]{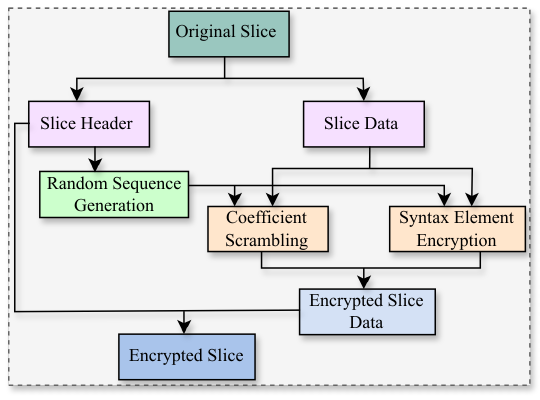}
\caption{Encryption process.}
	\label{fig:4}
\end{figure}
\subsection{Random Sequence Generation}
The AES-CTR is used to create a random sequence considering the initial key size of 256 bits and the initial vector size of 128 bits. The process for generating slice\_key and slice\_IV is explained in the following steps:

Step 1: Obtain the \textit{slice header} and input it into SHA-384 to generate a 384-bit slice\_header\_hash.

Step 2: Perform the \textit{XOR operation} on the first 256 bits of the slice\_header\_hash and the initial key to obtain slice\_key and perform the XOR operation on the last 128 bits of the slice\_header\_hash and the initial vector to obtain slice\_IV as computed in Eqs. (\ref{eq1a}) and (\ref{eq1b}), respectively.
\begin{gather}
	\label{eq1a}
slice\_key = slice\_header\_hash \oplus initial\;key \\
	\label{eq1b}
slice\_IV = slice\_header\_hash \oplus initial\;vector
\end{gather}
The process is repeated each time the encoder encodes a new slice to update slice\_key and slice\_IV. The newly generated slice\_key and slice\_IV are then used as inputs to the AES-CTR to produce a set of random sequences with a length of 64 bits. Additionally, the slice\_IV is incremented to generate a different random sequence for the next time.
\subsection{Coefficient Scrambling}
The coefficient scrambling process comprises a series of steps, as depicted in Fig. \ref{fig:5}, and is explained below.
\begin{figure}[!ht]
	\centering
	\includegraphics[width=1.05\columnwidth]{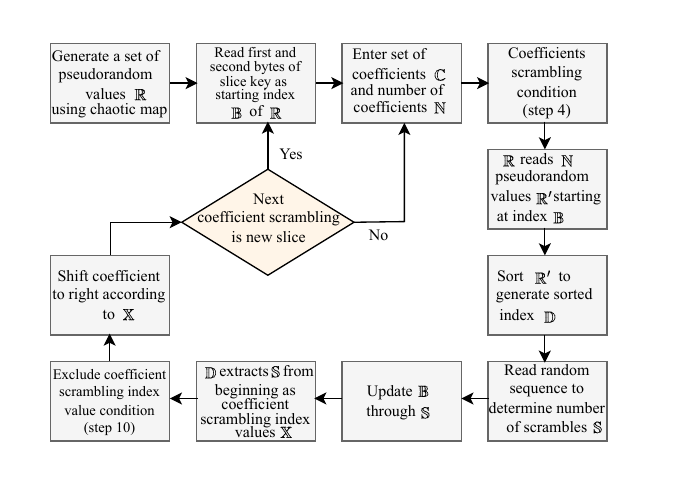}
\caption{Chaotic map-based coefficient scrambling process.}
	\label{fig:5}
\end{figure}
\begin{itemize}
\item Step 1: Generate a set of pseudorandom values $\mathbb{R}$ using chatoic map through Eq. (\ref{eq1}). The system combines the nonlinear dynamics of the logistic and sine maps to enhance ergodicity, sensitivity to initial conditions, and statistical randomness while avoiding the computational burden of higher-dimensional chaotic systems. The logistic-sine low-dimensional discrete chaotic system proposed in \cite{c16} is used to generate pseudorandom values to avoid excessive execution time. 
\begin{equation}
	\label{eq1}
	x(n+1)=\frac{4 \mu}{9} \ast x(n) \ast (1-x(n))+\frac{9-\mu}{9} \ast \sin [\pi\ast x(n)]
\end{equation}
where $0\leq \mu \leq 9$ and where $x(n)$ is a function used to generate pseudorandom values $\in \mathbb{R}$. It is assumed that $\mu =$ 7.3, $x(0) =$ 0.9324, and it iterates 65535 times, resulting in a total of 65536 pseudorandom values.

\item Step 2: Determine the starting index $\mathbb{B}$ of $\mathbb{R}$ by selecting the first and second bytes of slice\_key.

\item Step 3: Specify the coefficient group $\mathbb{C}$ and the total number of coefficients $\mathbb{N}$. If $\mathbb{N} \leq 1$, then end the coefficient scrambling process and jump to step 12 directly.

\item Step 4: Check the coefficient scrambling conditions. The conditions exclude the existence of coefficients in $\mathbb{C}$ as 1 and exclude the existence of the first coefficient greater than 1 in the front section of $\mathbb{C}$ as 2. The range for the front section is $1 \thicksim \min(n, 8)$; if it exists, set $\mathbb{F} = 1$. After that, the coefficients are eliminated; if $\mathbb{N} \leq 1$, end the coefficient scrambling process, and jump to step 12.

\item Step 5: Now, the pseudo-random value $\mathbb{R}$ reads $\mathbb{N}$ pseudorandom values $\mathbb{R}^{\prime}$ from index $\mathbb{B}$.

\item Step 6: Sort $\mathbb{R}^{\prime}$ to generate the sorted index $\mathbb{D}$. The process for sorting index generation is illustrated in Fig. \ref{fig:6}.
\begin{figure}[!htbp]
	\centering
	\includegraphics[width=0.9\columnwidth]{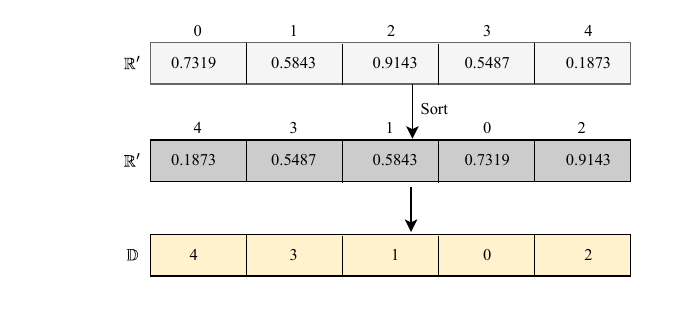}
\caption{Sorted index generation.}
	\label{fig:6}
\end{figure}

\item Step 7: Determine the number of coefficient scrambles $\mathbb{S}$ by using Eq. (\ref{eq2}), which needs to be associated with the random sequence value $\mathbb{S}_i$.
\begin{equation}
	\label{eq2}
	\mathbb{S}=\mod \left(\mathbb{S}_i, \mathbb{N}-1\right) + 2; \,\,\, 0\leq \mathbb{S}_i \leq 15
\end{equation}
\item Step 8: Update $\mathbb{B}$ through $\mathbb{S}$; here, set $\mathbb{B}=\mathbb{B}+\mathbb{S}$.

\item Step 9: $\mathbb{D}$ extracts $\mathbb{S}$ from the beginning as coefficient scrambling index values $\mathbb{X}$.

\item Step 10: If $\mathbb{F}$= 0, check the exclusion coefficient scrambling index value condition. The condition is whether $\mathbb{C}$ and $\mathbb{X}$ exist when $\mathbb{C}[\mathbb{X}[i-1]]$ = 2 and $\mathbb{X}[i]$ = 0. If this occurs, exclude $\mathbb{X}[i]$, $\mathbb{S} -$ 1. After the coefficient scrambling index values are eliminated, if $\mathbb{S} \leq$ 1, terminate the coefficient scrambling process and jump to step 12.

\item Step 11: Shift the coefficient to the right according to $\mathbb{X}$. The process of coefficient perturbation is depicted in Fig. \ref{fig:7}.
\begin{figure}[!ht]
	\centering
	\includegraphics[width=0.9\columnwidth]{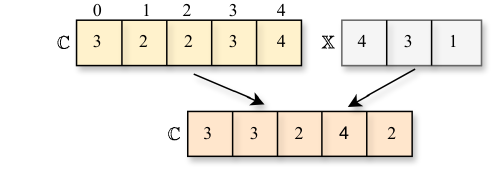}
\caption{Coefficient perturbation.}
	\label{fig:7}
\end{figure}

\item Step 12: If the next coefficient scrambling is a new slice, return to step 2; otherwise, return to step 3.
\end{itemize}

Compared with the coefficient scrambling scheme proposed by related research, the edge coefficient scrambling scheme proposed by the Scr scheme of \cite{c5-8746776} uses the exchange value method to perturb the coefficients. However, the proposed scheme is not associated with a random sequence and uses only the rand function to determine the exchange object, which results in extended security. The coefficient scrambling scheme of \cite{c13-9903590} uses a two-round shift for perturbation, as illustrated in Fig. \ref{fig:8}.
\begin{figure}[!ht]
	\centering
	\includegraphics[width=0.9\columnwidth]{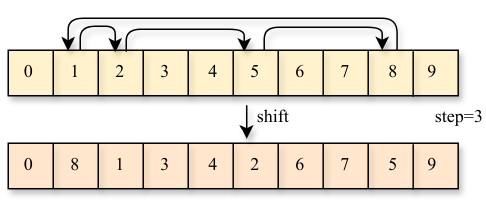}
\caption{One-round shift.}
	\label{fig:8}
\end{figure}

If the step is larger, the amplitude of the disturbance is smaller, and for a smaller step, it is equivalent to shifting all the coefficients to the right, and the degree of disturbance is still not good. Therefore, the coefficient scrambling scheme proposed by \cite{c13-9903590} is still limited in the distribution of coefficient perturbations. The scheme proposed in this paper successfully furnishes a more random coefficient perturbation distribution.
\subsection{Syntax Element Encryption} \label{sec:see}
\subsubsection{MVD Information}
The encryption operations MVD\_Hor\_Sign and MVD\_Ver\_Sign are performed via Eqs. (\ref{eq3}) and (\ref{eq4}), respectively.
\begin{equation}
	\label{eq3}
	\begin{split}
	en\_MVD\_Hor\_Sign= \\
	MVD\_Hor\_Sign \oplus \mathbb{S}_i;\,0\leq \mathbb{S}_i\leq 1
	\end{split}
	\end{equation}
\begin{equation}
	\label{eq4}
	\begin{split}
	en\_MVD\_Ver\_Sign= \\
	MVD\_Ver\_Sign \oplus \mathbb{S}_i;\, 0\leq \mathbb{S}_i \leq 1
	\end{split}
\end{equation}
The MVD size consists of the base level and the remaining level. The structure of the remaining level is the [$\mathbb{N}$-1 zeros prefix][1][$\mathbb{N}$ bits suffix]. The remaining\_MVD\_Hor\_Suffix and remaining MVD\_Ver\_Suffix are encrypted via Eqs. (\ref{eq5}) and (\ref{eq6}), respectively.
\begin{equation}
		\label{eq5}
	\begin{split}
	en\_remaining\_MVD\_Hor\_Suffix= \\
	remaining\_MVD\_Hor\_Suffix \oplus \mathbb{S}_i;\,0\leq \mathbb{S}_i\leq 2^{\mathbb{N}} 
	\end{split}
\end{equation}
\begin{equation}	
	\label{eq6}
	\begin{split}
	en\_remaining\_MVD\_Ver\_Suffix=\\
	 remaining\_MVD\_Ver\_Suffix \oplus \mathbb{S}_i;\, 0\leq \mathbb{S}_i \leq 2^{\mathbb{N}}
	\end{split}
\end{equation}
\subsubsection{QTC Information}
The encryption coeff\_Sign is performed via Eq. (\ref{eq7}).
\begin{equation}
	\label{eq7}
		en\_coeff\_Sign= coeff\_Sign \oplus \mathbb{S}_i;\,0\leq \mathbb{S}_i\leq 1 
\end{equation}
The QTC size consists of the base level and the remaining level. The structure of the remaining level is [$\mathbb{M}$-1 ones prefix][0][$\mathbb{M}$ bits suffix].

The rice parameter $r$ of the TR code is computed via Eq. (\ref{eq8}).
\begin{equation}
	\label{eq8}
	r= \min (4, r+1) \,\,\, if\,\,\,  QTC\_size >3\times2^r 
\end{equation}
Encryption of the remaining\_Coeff\_Suffix may change the rice parameter, causing format-compliant failure. Therefore, the remaining\_Coeff\_Suffix requires certain conditions before encryption is allowed, as does the number of encryption bits $\mathbb{L}$. If encryption is allowed, the remaining\_Coeff\_Suffix is encrypted via Eq. (\ref{eq9}).
\begin{equation}
	\label{eq9}
	\begin{split}
	en\_remaining\_Coeff\_Suffix=\\
	 remaining\_Coeff\_Suffix \oplus \mathbb{S}_i;\, 0\leq \mathbb{S}_i \leq 2^{\mathbb{L}}
	 \end{split}
\end{equation}
\subsubsection{Luma IPM Information}
The Luma IPM comprises 35 prediction modes, including planar mode, DC mode, and 33 angle modes. Encryption luma\_IPM is performed via Eq. (\ref{eq10}).
\begin{equation}
	\label{eq10}
	\begin{split}
		en\_luma\_IPM=\\
		\mod (luma\_IPM + \mathbb{S}_i, 35);\, 0\leq \mathbb{S}_i \leq 63
	\end{split}
\end{equation}
Directly encrypting luma\_IPM will cause format-compliant failure; therefore, luma\_IPM and en\_luma\_IPM are first recorded, followed by three additional actions.
\begin{enumerate}
\item To determine the three most likely candidate prediction modes of adjacent PUs, en\_luma\_IPM is used to obtain the luma\_IPM of adjacent PUs.
\item The coefficient scan mode is associated with luma\_IPM \cite{c12-7370952} and is determined via en\_luma\_IPM.
\item Chroma\_IPM is also associated with luma\_IPM, as depicted in Table \ref{tab:fig9}.
\end{enumerate}
\begin{table}[!htbp]
	\caption{Association of chroma\_IPM and luma\_IPM}\label{tab:fig9}
	\centering
	\resizebox{0.98\columnwidth}{!}{
		\tiny
		\begin{tabular}{|p{15pt}|p{31pt}|c|c|c|c|c|}
			\hline 
				\multicolumn{2}{|c|}{} & \multicolumn{5}{c|}{\textbf{Corresponding luma prediction mode}} \\
				\cline{3-7}\multicolumn{2}{|c|}{\textbf{Chroma IPM ID}}& \textbf{0} & \textbf{26}& \textbf{10}& \textbf{1} & \textbf{34} \\  
				\multicolumn{2}{|c|}{}& \textbf{(Planar)} & \textbf{(Vertical)}& \textbf{(Horizontal)}& \textbf{ (DC)} &  \\  \hline 
				 \cline{2-7}&0 (Planar)&4&0&0&0&0 \\ 
				\cline{2-7}Chroma &26 (Vertical)&1&4&1&1&1 \\ 
				\cline{2-7}prediction&10 (Horizontal)&2&2&4&2&2 \\ 
				\cline{2-7}mode&1 (DC)&3&3&3&4&3 \\ 
				\cline{2-7}&34&0&1&2&3&4 \\ \hline
	\end{tabular}}
	
\end{table}

If the corresponding luma\_IPM is encrypted and the chroma\_IPM ID changes, a different chroma\_IPM is obtained during decoding. Therefore, before encoding chroma\_IPM, we must first determine whether chroma\_IPM needs to be adjusted; as shown in Algorithm \ref{algo:alg1}. Furthermore, when the chroma\_IPM ID is determined, the corresponding luma\_IPM uses en\_luma\_IPM. Additionally, the coefficient scan mode is also associated with chroma\_IPM, and the adjusted chroma\_IPM is used to make final decisions. 
\section{Theoretical Analysis}\label{sec: ta}

\subsection{Formal Theoretical Proof to Resist Chosen-Plaintext Attacks} \label{sec:ta1}
\begin{theorem}
	The proposed Syntax Element Encryption for H.265/HEVC using Chaotic Map-Based Coefficient Scrambling Scheme is secure against Chosen-Plaintext Attack (CPA) under the assumption that the underlying chaotic map is computationally indistinguishable from a random permutation and the SHA-384 based key update process behaves as a pseudo-random function (PRF). 
\end{theorem}

\begin{proof}
	Under the \textit{Chaotic Map Pseudorandomness Assumption (CMPA)} the output sequence $S$ ( initialized with a secret key $K$) is indistinguishable from a truly random permutation without knowledge of $K$. Assume for contradiction there exists a probabilistic polynomial-time (PPT) adversary $\mathcal{A}$ breaks the CPA security of our scheme with non-negligible advantage $\epsilon$. We construct a simulator $\mathcal{B}$ that uses $\mathcal{A}$ to distinguish between a real chaotic map and a random permutation, thereby breaking CMPA or PRF assumptions:
		\begin{itemize}
			\item $\mathcal{B}$ receives a real chaotic permutation derived from a secret key $K$, or a truly random permutation.
			\item $\mathcal{A}$ selects two plaintext coefficient sets $C_0$ and $C_1$.
			\item $\mathcal{B}$ selects a random bit $b \in \{0,1\}$, scrambles $C_b$ using the permutation provided by $\mathcal{C}$, and returns the ciphertext $C_b'$ to $\mathcal{A}$. Here, $\mathcal{C}$ is the plaintext syntax element set being scrambled and encrypted. 
			\item If $b' = b$, then $\mathcal{B}$ guesses that the permutation is real chaotic map; otherwise, random.
		\end{itemize}
	If $\mathcal{A}$ has advantage $\epsilon$, so does $\mathcal{B}$, violating CMPA or PRF assumptions. Since these primitives are secure by hypothesis, $\epsilon$ must be negligible. Thus, the proposed scheme resists CPA.
\end{proof}
\subsection{Chaotic System Dynamics Analysis} \label{sec:ta2}
\begin{theorem}
	The Logistic-Sine Chaotic Map (LSCM) used in the proposed Syntax Element Encryption for H.265/HEVC exhibits properties of a robust chaotic system namely, parameter sensitivity, chaotic region validity, and positive Lyapunov exponent ensuring the unpredictability and security of the generated pseudorandom sequences.
\end{theorem}

\begin{proof}
	Consider the LSCM, as defined in Eq. (\ref{eq1}), where $x(n) \in (0,1)$ is the chaotic state at iteration $n$, and $\mu \in [0,9]$ is the control parameter. 
		\begin{itemize}
			\item {Parameter Sensitivity: The LSCM exhibits strong sensitivity to initial conditions and parameters. For two nearby initial states $x_0$ and $x_0 + \delta$, the divergence satisfies below condition: 
					\[
						|x_n - x_n'| \approx e^{\lambda n} |\delta|,
					\]
			where $\lambda$ denotes the Lyapunov exponent. A positive $\lambda$ implies that small differences in $x_0$ or $\mu$ result in exponentially divergent sequences, ensuring unpredictability of the generated pseudorandom numbers.}
		\item {Chaotic Region Validation: To ensure fully chaotic behavior, the control parameter $\mu$ is selected within a specific range that maintains non-periodicity and complexity of the output. Accordingly, the system exhibits strong chaotic characteristics for $\mu \in [6,9]$. Additionally, the initial condition $x_0 \in (0,1)\setminus\{0,1\}$ prevents trivial orbits and degeneration into fixed points or short cycles.}
		\item {Lyapunov Exponent Analysis: The system's chaoticity is quantified by the Lyapunov exponent defined as 
		\[
			\lambda = \lim_{n \to \infty} \frac{1}{n} \sum_{i=0}^{n-1} \ln \left| f'(x_i) \right|,
		\]	
		where the derivative $f'(x)$ of the map (in Eq. (\ref{eq1})) is defined as: 
		\[
			f'(x) = \frac{4 \mu}{9} \cdot (1 - 2x) + \frac{9 - \mu}{9} \cdot \pi \cdot \cos(\pi x).
		\]
		For $\mu \in [6,9]$ within the chaotic region, numerical evaluation yields $\lambda > 0$, confirming the system’s chaotic nature and suitability for cryptographic applications.}
		\end{itemize}
		This analysis verifies that the LSCM provides high sensitivity, validated chaotic behavior, and positive Lyapunov exponent ensuring secure pseudorandom sequence generation for coefficient scrambling in the proposed H.265/HEVC syntax element encryption scheme.	
\end{proof}
\subsection{Theoretical Proof for Information Entropy Analysis}\label{sec:ta3}
\begin{theorem}
	The proposed Chaotic Map-Based Coefficient Scrambling Scheme for H.265/HEVC significantly increases the information entropy of syntax elements, resulting in enhanced randomness and resistance against statistical and entropy-based attacks.
\end{theorem}

\begin{proof}
	For a discrete random variable $X$ with possible outcomes $\{x_1, x_2, \dots, x_n\}$ and corresponding probabilities $P(X) = \{p_1, p_2, \dots, p_n\}$, the Shannon entropy $H(X)$ is computed using Eq. (\ref{eq18}). 
	\begin{itemize}
		\item {Pre-Encryption Entropy: 	Let $C = \{c_1, c_2, \dots, c_n\}$ represent the set of residual coefficients before encryption (i.e. original), following a Laplacian-like distribution:
		\[
		P(c) = \frac{1}{2b} \exp\left(-\frac{|c|}{b}\right)
		\]
		which implies a higher probability of small coefficient values (0, $\pm$1), and lower probabilities for larger values. This distribution is peaked around zero, resulting in lower entropy. Therefore, the entropy of $C$ satisfies:
		\[
		H(C) < H_{\max}= \log_2 n
		\]
			For 8-bit coefficients ($n=256$), $H_{\max} = 8$ bits/symbol. Empirical HEVC data yields $H(C) \approx 5.5$ bits/symbol i.e. commonly in range of 5 to 6 bits/symbol.
		}
		\item {Post-Encryption Entropy: After applying the chaotic map-based scrambling, two possible scenarios are:
		\begin{itemize}
			\item {\textit{Permutation-only Scrambling:} Histogram remains unchanged; Entropy $H(C') \approx H(C)$; Spatial correlation destroyed}.
			
			\item {\textit{Chaotic Value Masking:} 
				Each coefficient is modified as: 
			\[
			c_i' = c_i \oplus k_i
			\]
			where $k_i$ is drawn from a pseudo-random sequence generated by the Logistic-Sine Chaotic Map (LSCM). This produces an approximately uniform distribution:
			\[
			P'(x) \approx \frac{1}{2^b}, \quad b=8
			\]
			yielding:
			\[
			H(C') \approx 8 \ \text{bits/symbol}.
			\]}
		\end{itemize} }
	\item {Theoretical Entropy Gain: The entropy gain $\Delta H$ satisfies:
		\[
		\Delta H = H(C') - H(C) \approx 8.0 - 5.5 = 2.5 \ \text{bits/symbol},
		\]
		demonstrating a significant increase in randomness and unpredictability.}		
	\end{itemize}
	Due to this entropy enhancement, the encrypted coefficients closely resemble uniformly random data, reducing the risk of entropy-based and statistical attacks. Thus, the scheme improves security from an information-theoretic perspective.
\end{proof}
\subsection{Security Proof}\label{sec:ta4}
\begin{theorem}
	The proposed Syntax Element Encryption for H.265/HEVC using Logistic-Sine Chaotic Map (LSCM) and SHA-384 key update mechanism ensures information-theoretic security, including large key space, high entropy, strong confusion and diffusion, and resistance against Chosen-Plaintext Attacks (CPA).
\end{theorem}

\begin{proof}
	Key Space: With $x_0 \in (0,1)$ and $\mu \in [6,9]$ at $10^{-15}$ precision, and SHA-384 key updates ($2^{384}$ possibilities), the total key space exceeds $2^{300}$, far above the $2^{128}$ security benchmark. \\
	Entropy: Encrypted syntax elements approach a uniform distribution:
		\[
		H(X) \approx \log_2 N,
		\]
		ensuring high unpredictability and randomness. \\
	Confusion and Diffusion: LSCM provides nonlinearity (confusion), while dynamic SHA-384 key updates propagate changes across syntax elements (diffusion), ensuring the avalanche effect. \\
	CPA Resistance: Assuming LSCM outputs are computationally indistinguishable from random and SHA-384 behaves as a pseudo-random function, the scheme achieves security. \\
	Mutual Information:  
		\[
		I(M;C) = H(M) - H(M|C) \approx 0,
		\]
		implying ciphertext reveals negligible information about plaintext. \\
	The scheme satisfies key cryptographic criteria: large key space, high entropy, strong confusion and diffusion, and CPA resistance.
\end{proof}
\section{Operational Design and Performance Metrics} \label{sec:opd}
\subsection{Operational Design}
The operational summary for the proposed scheme is described in Algorithm \ref{algo:alg1}.

\begin{algorithm}[!ht]
\caption{Proposed Scheme: Operational summary}
	\label{algo:alg1}
\textbf{Input} chroma\_IPM, luma\_IPM, and en\_luma\_IPM \\
\textbf{Output} updated chroma\_IPM \\
\textbf{Begin:} \\
\If{\text{chroma\_IPM}=34 \text{and} \text{luma\_IPM} $\ne$ \text{en\_luma\_IPM}}{chroma\_IPM=luma\_IPM}
\ElseIf{chroma\_IPM = en\_luma\_IPM}{chroma\_IPM=34}
\Else{no action}
\textbf{end if} \\
\textbf{End}
\end{algorithm}
\subsection{Performance Metrics}
The two indicators peak signal-to-noise ratio (PSNR) and structural similarity (SSIM) are computed for further evaluation. The PSNR value is computed via Eq. (\ref{eq11}).
\begin{equation}
	\label{eq11}
	PSNR(\mathbb{X},\mathbb{Y})=10\log_{10}\frac{(2^{\mathcal{N}}-1)^2}{MSE}
\end{equation}
where $\mathcal{N}$ represents the grayscale level of the image. This scheme assumes that $\mathcal{N}=$ 8, which is an 8-bit grayscale image. The mean square error (MSE) is expressed by Eq. (\ref{eq12}).
\begin{equation}
	\label{eq12}
	MSE= \frac{1}{\mathcal{H}\times\mathcal{W}}\sum_{i=1}^{\mathcal{H}}\sum_{j=1}^{\mathcal{W}}\biggl(\mathbb{X}(i,j)-\mathbb{Y}(i,j)\biggr)^2
\end{equation}
where $\mathcal{H}$ and $\mathcal{W}$ define the height of the image and the width of the image, respectively. The PSNR is measured in dB. A lower PSNR value indicates a lower similarity between the two frames. The SSIM values are computed via Eq. (\ref{eq13}).
\begin{equation}
	\label{eq13}
	SSIM(\mathbb{X},\mathbb{Y})=\frac{(2\mu_{\mathbb{X}}\mu_{\mathbb{Y}}+{\mathbb{C}_1})(2\sigma_{\mathbb{X}\mathbb{Y}}+{\mathbb{C}_2})}{(\mu_{\mathbb{X}}^2+\mu_{\mathbb{Y}}^2+{\mathbb{C}_1})(\sigma_{\mathbb{X}}^2+\sigma_{\mathbb{Y}}^2+{\mathbb{C}_2})}
\end{equation}
where ${\mathbb{C}_1}=(K_1L)^2$, ${\mathbb{C}_2}=(K_2L)^2$, the default value of $K_1$ is 0.01, the default value of $K_2$ is 0.03, and $L$ denotes the maximum grayscale value of the image; for an 8-bit grayscale image, $L$ is 255. $\mu_{\mathbb{X}}$, $\sigma_{\mathbb{X}}$, and $\sigma_{\mathbb{X}\mathbb{Y}}$ represent the mean, standard deviation, and covariance, respectively, as computed via Eqs. (\ref{eq14}-\ref{eq16}).
\begin{gather}
	\label{eq14}
\mu_{\mathbb{X}}=\frac{1}{m}\sum_{i=1}^{m}\mathbb{X}_i \\
	\label{eq15}
\sigma_{\mathbb{X}}= \frac{1}{m-1}\sum_{i=1}^{m} \sqrt{ \biggl(\mathbb{X}_i - \mu_{\mathbb{X}} \biggr)^2} \\
	\label{eq16}
\sigma_{\mathbb{X}\mathbb{Y}}= \frac{1}{m-1}\sum_{i=1}^{m} \biggl(\mathbb{X}_i - \mu_{\mathbb{X}} \biggr)\biggl(\mathbb{Y}_i - \mu_{\mathbb{Y}} \biggr)
\end{gather}
where $m$ is the total number of pixels in the image. The computations for $\mu_{\mathbb{Y}}$ and $\sigma_{\mathbb{Y}}$ are similar to those for $\mu_{\mathbb{X}}$ and $\sigma_{\mathbb{X}}$. The SSIM value ranges between [0, 1]. A smaller value indicates a lower similarity between the two frames.

Another important metric, the edge difference ratio (EDR) indicator required for performance evaluation, is computed via Eq. \eqref{eq17}.
\begin{equation}
	\label{eq17}
	EDR(\mathbb{X},\mathbb{Y})= \frac{\sum_{i=1}^{\omega}|E_{\mathbb{X}}(i)-E_{\mathbb{Y}}(i)|}{\sum_{i=1}^{\omega}|E_{\mathbb{X}}(i)+E_{\mathbb{Y}}(i)|}
\end{equation}
where $\omega$ represents the total number of pixels in the edge frame. The value of the EDR ranges from 0 to 1. A higher value indicates a greater difference between the two edge frames.

The encryption operation may increase the file size of the encrypted video, so the bit rate overhead is also one of the considerations when choosing to encrypt the video. The formula for bit rate overhead is expressed in Eq. \eqref{eqbro}.
\begin{equation}
	\label{eqbro}
	BRO= \frac{ES-OS}{OS}\times 100\%
\end{equation}
where $ES$ is the file size of the encrypted video and where $OS$ is the file size of the original video. The unit of bitrate overhead is the percentage ($\%$), and the ideal value is $0\%$, which means a constant bitrate.

The average Hamming distance ($HD_{avg}$) to calculate the difference between two sets of random sequences and to evaluate the resistance to chosen plaintext attacks is computed via Eq. \eqref{eq:pta}.
\begin{equation}
	\label{eq:pta}
	HD_{avg}= \frac{\sum_{i=1}^{K_L}\mathbb{S}_x(i)\oplus \mathbb{S}_y(i)}{K_L}
\end{equation}
where $K_L$ represents the minimum length of the two sets of random sequences and where $\mathbb{S}_x (i)$ and $\mathbb{S}_y (i)$ represent the ith values of the first and second sets of random sequences, respectively. The value of $HD_{avg}$ is between 0 and 1. The closer its value is to 0.5, the greater the difference between the two sets of random sequences is, and it can effectively resist chosen plaintext attacks.

The information entropy is computed via Eq. \eqref{eq18}.
\begin{equation}
	\label{eq18}
	\mathbb{H}(\mathbb{X})=-\sum_{i=0}^{\mathcal{N}}P(\mathbb{X}_i)\times\log_2P(\mathbb{X}_i)
\end{equation}
where $\mathcal{N}$ represents the maximum grayscale value of the image. For an 8-bit grayscale image, $\mathcal{N}$ is 255. $P(\mathbb{X}_i)$ represents the probability of occurrence of frame $\mathbb{X}$ when the pixel value is $i$. For an 8-bit grayscale image, the maximum value of information entropy is 8. A value of information entropy closer to 8 reflects a higher degree of randomness. 

Furthermore, widely used to test the robustness of image cryptosystems against differential cryptanalysis, the number of pixels changing rate (NPCR) and unified average change intensity (UACI) are computed via Eqs. (\ref{eq19}--\ref{eq21}).
\begin{gather}
	\label{eq19}
NPCR= \frac{\sum_{u, v} O(u, v)}{\mathcal{H}\times\mathcal{W}}\times 100\% \\
	\label{eq20}
O(u, v)= \begin{cases}
0, \textit{if} \quad C_1(u, v) = C_2(u,v) \\
1, \textit{if}\quad C_1(u, v) \neq C_2(u,v)
\end{cases} \\
	\label{eq21}
UACI= \frac{1}{\mathcal{H}\times\mathcal{W}}\bigl[\sum_{u, v} \frac{C_1(u, v) - C_2(u,v)}{2^L-1} \bigr]\times 100\%
\end{gather}
where $ C_1(u, v)$ and $C_2(u, v)$ represent the pixel values of two cipher frames at position $(u, v)$. The expected values of the NPCR and UACI are computed via Eqs. \eqref{eq22} and \eqref{eq23}, respectively.
\begin{gather}
	\label{eq22}
NPCR^{\dagger}= (1-2^{-L})\times100\%  \\
		\label{eq23}
UACI^{\dagger}= \frac{1}{2^{2L}}\frac{\sum_{u=1}^{2^L-1}u(u+1)}{2^L-1}\times100\%
\end{gather}
where $L=$8 is the length of the pixel values. 

\section{Performance evaluation} \label{sec:res}
\subsection{Experimental Setup}
The experimental works are conducted on a machine comprising an environment configuration, as depicted in Table \ref{tab:1}. The profile used by HM17.0 \cite{c20} is encoder\_randomaccess\_main, in which the QP is set to 24, the number of \textit{encoding frames} (FramesToBeEncoded) is set to 100, and \textit{SideHideFlag} is set to false. Additionally, \textit{AES initial key} is set to \{0x04, 0xB0, 0x21, 0x90, 0xCB, 0xA6, 0xAF, 0xF6, 0x94, 0xFD, 0xAA, 0x9B, 0xEB, 0xEE, 0xA6, 0xCD, 0xF1, 0xF2, 0xF1, 0xBC, 0xCE, 0x17, 0x45, 0xE2, 0x39, 0x57, 0x68, 0x8A, 0xCB, 0x8A, 0x21, 0x5A\}, a total of 256 bits, the \textit{initial vector} is set to \{0x17, 0x3B, 0xFC, 0xD8, 0x7B, 0x33, 0x0D, 0x30, 0x0D, 0xDB, 0x12, 0x65, 0xAF, 0xA9, 0x47, 0x52\}, with a total of 128 bits.
\begin{table}[!htbp]
	\caption{Experimental environment configuration}\label{tab:1}
	\centering
	\resizebox{0.9\columnwidth}{!}{
		\tiny
		\begin{tabular}{|p{50pt}|p{70pt}|}
			\hline 
			\textbf{Item} & \textbf{Specification} \\ \hline 
			Operating System &	Microsoft Windows 10 64bit \\  \hline 
			CPU &	Intel(R) Core(TM) i7-4790 CPU @ 3.60GHz \\  \hline 
			RAM	& 24.0 GB \\  \hline 
			Software &	Microsoft Visual Studio 2019, MATLAB \\  \hline 
			Programming Language &	C++ \\  \hline 
			Library	& Crypto++ \\  \hline 
			H.265 Codec & HM 17.0 \\  \hline
	\end{tabular}}
	
\end{table}

The simulation and experiments employed a 12-video sequence with resolutions ranging from 352$\times$288 to 2560$\times$1600 \cite{c21}, \cite{c22}, as shown in Table \ref{tab:2}, and were compared with the robust RSVE scheme \cite{c13-9903590}.
\begin{table}[!htbp]
	\caption{Video sequence}\label{tab:2}
	\centering
	\resizebox{0.8\columnwidth}{!}{
		\tiny
		\begin{tabular}{|l|c|c|}
			\hline 
			\textbf{Video} & \textbf{Resolution} & \textbf{Frame Rate}\\ \hline 
			Mobile & 352$\times$288 & 30fps\\  \hline 
			Foreman	&352 $\times$ 288	&30 fps \\  \hline 
			BlowingBubbles	& 416 $\times$ 240	& 50 fps \\  \hline 
			RaceHorses	&416 $\times$240&	30 fps\\  \hline 
			PartyScene	&832$\times$480&	50 fps\\  \hline 
			BasketballDrillText	&832 $\times$ 480&	60 fps \\  \hline 
			Johnny	&1280$\times$720	&60 fps \\  \hline 
			KristenAndSara	&1280$\times$720	&60 fps \\  \hline 
			BasketballDrive	&1920$\times$1080	&50 fps \\  \hline 
			BQTerrace&	1920 $\times$1080&	60 fps \\  \hline 
			Traffic	&2560$\times$1600&	30 fps \\  \hline 
			PeopleOnStreet	&2560$\times$1600&	30 fps \\  \hline 
	\end{tabular}}
	
\end{table}
\subsection{Experimental Results}
Fig. \ref{figexp} shows that compared with the original frame (column 1st), the encrypted frame (column 2nd) has difficulty distinguishing the original authenticity in terms of texture and outline. In addition, the decrypted frames (column 3rd) are identical to the original frames. It indicates the proposed scheme can significantly affect the subjective vision of the video content. 
\begin{figure*}[!ht]
	\centering
	\includegraphics[width=0.8\textwidth]{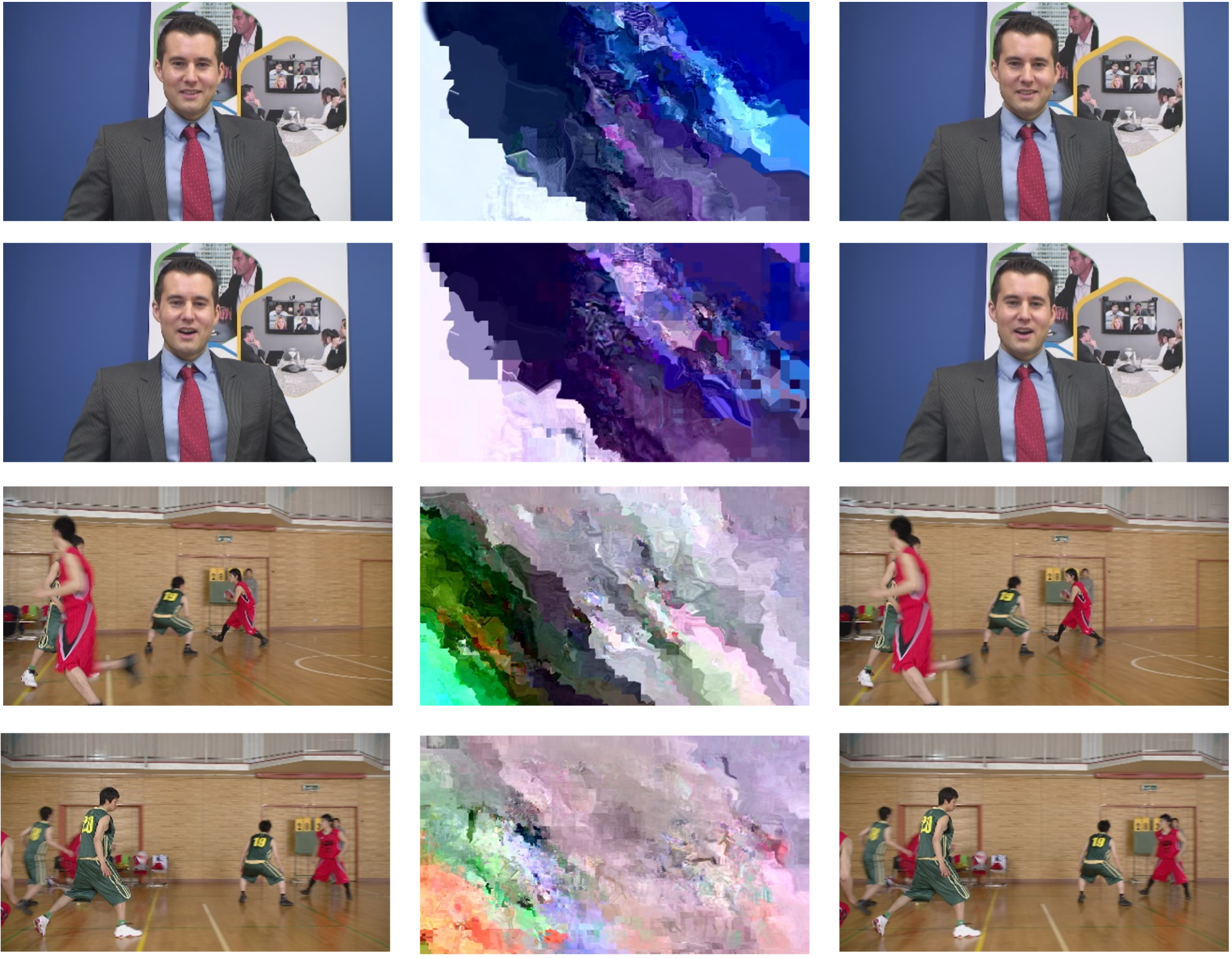}
\caption{Experimental results for the Johnny and BasketballDrive video sequences. The first and second rows are the Johnny video sequence (\#1 and \# 30 frame), and the third and fourth rows are the BasketballDrive video sequence (\#1 and \# 30 frame) with the proposed scheme.}
	\label{figexp}
\end{figure*}
\subsection{Slice Analysis}
Fig. \ref{fig:slice} shows the slice analysis for the Foreman video sequence. The video loses the 11th to 14th slices, and the decrypted content of the 15th frame is the same as that of the original frame (18th frame). The first and second columns are the results of the Foreman video in the 1st frame (the corresponding frame is also the 1st frame) and the 15th frame (the corresponding frame is the 18th frame). Starting from the 1st frame as a normal decrypted frame, some slices are lost along the way, and the subsequent decrypted frames are still exactly the same as the original corresponding frames. 
\begin{figure}[!ht]
	\centering
	\includegraphics[width=0.9\columnwidth]{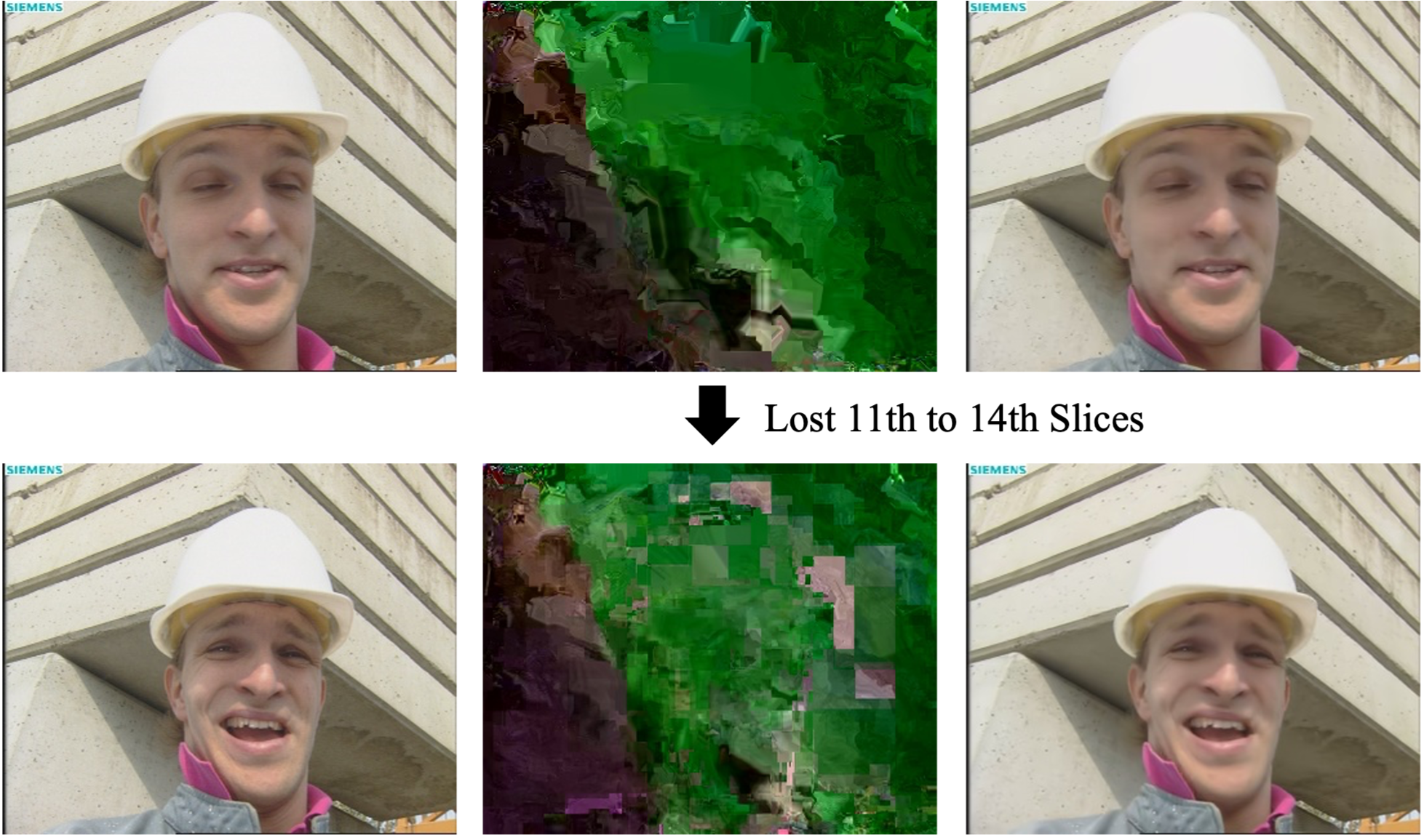}
\caption{Slice analysis: Foreman video original frame, encrypted frame and decrypted frame.}
	\label{fig:slice}
\end{figure}

Furthermore, the PSNR and SSIM indicators for the Foreman video in Table \ref{tab:slice} show that the 15th decrypted frame and the original corresponding frame (18th frame) are the same. Thus, the random sequence generated by the proposed scheme is synchronized with the slice. 
\begin{table}[!htbp]
	\caption{Slice analysis: PSNR and SSIM results}\label{tab:slice}
	\centering
	\resizebox{1.0\columnwidth}{!}{
			\begin{tabular}{|l|c|c|c|c|c|c|}
				\hline 
				\textbf{Video} &\multicolumn{3}{c|}{\textbf{PSNR}} & \multicolumn{3}{c|}{\textbf{SSIM}} \\ 
				\cline{2-7} \textbf{Frame} &\textbf{Original} &\textbf{Encrypt} &\textbf{Decrypt} &\textbf{Original} &\textbf{Encrypt} &\textbf{Decrypt} \\ \hline 
				1 (1)	&40.1872&	5.8797	&40.1872	&0.9633	&0.2207&	0.9633 \\ \hline
				15 (18)&	37.7070	&6.6155	&37.7070&	0.9477	&0.2138	&0.9477 \\ \hline
	\end{tabular}}
\end{table}
\subsection{Subjective Vision Analysis}
The outcomes for the subjective vision analysis are presented in Fig. \ref{fig:11}. The proposed method significantly impacts the subjective vision of the video content compared with the original frame. Additionally, the decrypted frame is identical to the original frame, confirming the correctness of the decryption. Furthermore, the proposed method causes the same level of distortion as the RSVE scheme does, indicating that it does not affect the subjective distortion levels of the video content.
\begin{figure*}[!ht]
	\centering
	\includegraphics[width=0.95\textwidth]{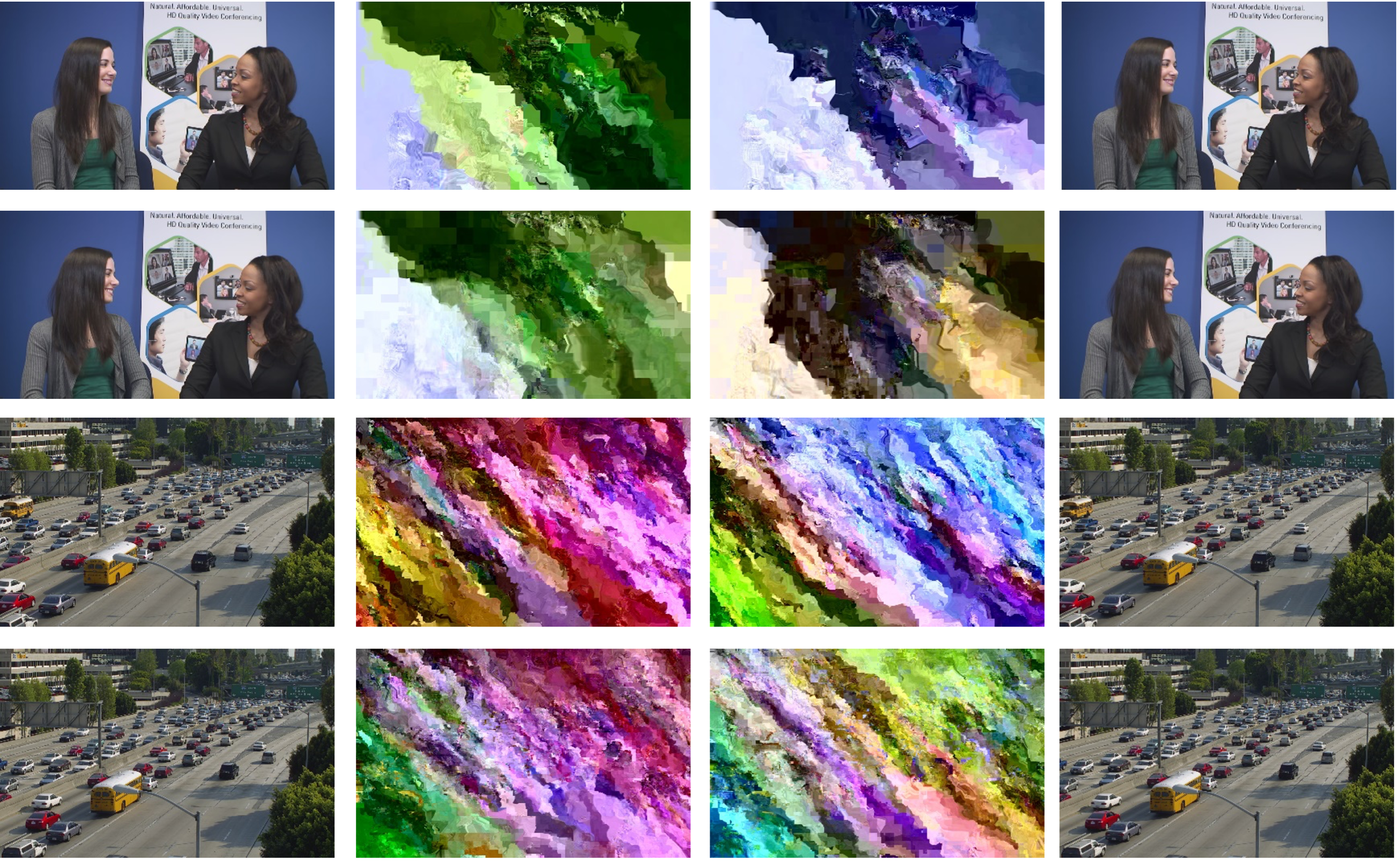}
\caption{Experimental results for subjective vision analysis. The first and second rows are the KristenAndSara video sequence (\#1 and \# 30 frame), and the third and fourth rows are the Traffic video sequence (\#1 and \# 30 frame). The first to fourth columns are the original frame, the RSVE scheme, the proposed scheme, and the decrypted frame, respectively.}
	\label{fig:11}
\end{figure*}
\subsection{Objective Indicator Analysis}
The evaluation uses two indicators: the peak signal-to-noise ratio (PSNR) and structural similarity (SSIM). As shown in Table \ref{tab:3}, all the results are averaged over 100 frames for PSNR and SSIM. Notably, the proposed scheme shows significantly lower PSNR values and performs second in SSIM when compared to the state-of-the-art methods like Jovanovic et al. \cite{c24}, Sallam et al. \cite{c23-8119905}, and the RSVE scheme of Chan et al. \cite{c13-9903590}. This indicates that the proposed scheme can produce the best-distorted image, but the RSVE method provides better outcomes in terms of visual perception. 
\begin{table*}[!htbp]
	\caption{PSNR and SSIM results for various encryption algorithms}\label{tab:3}
	\centering
	\resizebox{0.99\textwidth}{!}{
		\begin{tabular}{|l|c|c|c|c|c|c|c|c|c|c|c|}
			\hline 
			\multirow{3}{*}{\textbf{Video Sequence}} & \multirow{3}{*}{\textbf{QP}} &\multicolumn{5}{c|}{\textbf{PSNR}} &\multicolumn{5}{c|}{\textbf{SSIM}}\\ 
			\cline{3-12} &&\multirow{2}{*}{\textbf{Original}} &\textbf{Jovanović} &\textbf{Sallam} &\textbf{RSVE}&\multirow{2}{*}{\textbf{Proposed}} &\multirow{2}{*}{\textbf{Original}} &\textbf{Jovanović} &\textbf{Sallam} &\textbf{RSVE}&\multirow{2}{*}{\textbf{Proposed}} \\ 
			&&&\textbf{\cite{c24}} &\textbf{\cite{c23-8119905}} &\textbf{\cite{c13-9903590}}& &&\textbf{\cite{c24}} &\textbf{\cite{c23-8119905}} &\textbf{\cite{c13-9903590}}& \\
			\hline 
			\cline{3-12}\multirow{3}{*}{Mobile} &12 &51.2080 &10.8199 &10.1993 &9.9067	&\textbf{9.3520} &0.9984 &0.0948 &0.0741 &\textbf{0.0592} &0.0819  \\
			\cline{3-12} &24 &40.9997 &11.3320 &11.0287 &9.9786 &\textbf{9.9557} &0.9889 &0.0957 &0.1092 &\textbf{0.0663} &0.0914 \\
			\cline{3-12} &36 &32.9772 &10.2106 &10.4236 &\textbf{9.5058} &10.4717 &0.9313 &0.1098 &0.1446 &\textbf{0.0738} &0.1095 \\  \hline
			\cline{3-12}\multirow{3}{*}{Foreman} &12 &51.5915 &12.2185 &13.4685 &11.1852 &\textbf{10.5552} &0.9961 &0.3556 &0.3765 &\textbf{0.2381} &0.2860 \\
			\cline{3-12} &24 &44.4308 &13.8540 &11.9037 &12.4377  &\textbf{7.9867}	&0.9712 &0.3005 &0.3830 &\textbf{0.2279} &0.2514 \\
			\cline{3-12} &36 &38.4626 &13.0729 &13.8619 &12.7978  &\textbf{8.3192} &0.9015 &0.4150 &0.4163 &0.3437 & \textbf{0.2879}\\  \hline
			\cline{3-12}\multirow{3}{*}{BlowingBubbles} &12 &51.0548& 13.1237 &13.5211 &\textbf{10.7650} & 11.1731	&0.9977 &0.1812 &0.1887 &\textbf{0.1693} &0.1935 \\
			\cline{3-12}&24	&41.5218 &13.0176 &13.9822 &11.4384	&\textbf{10.9450} &0.9701 &0.1446 &0.2549 &\textbf{0.1697} & 0.1946 \\
			\cline{3-12} &36 &34.7572 &12.2965 &12.5875 &10.7501 &\textbf{10.4957} &0.8548 &0.2278 &0.2706 &\textbf{0.1648} &0.2085 \\  \hline
			\cline{3-12}\multirow{3}{*}{RaceHorses}	&12	&51.3729 &11.2945 &11.9913 &\textbf{10.0016} &10.8395 &0.9975 &0.2124 &0.2068 &\textbf{0.1749} &0.2097 \\
			\cline{3-12}&24	&42.0433 &11.9281 &11.0942 &\textbf{10.7121} &11.1435 &0.9815 &0.2025 &0.2412 &\textbf{0.1730}  &0.1991 \\
			\cline{3-12} &36 &34.7851 &12.0602 &11.5157 &10.8476 &\textbf{9.3066} &0.8797 &0.2498 &0.2985 &\textbf{0.1797} &0.2286 \\  \hline
			\cline{3-12}\multirow{3}{*}{PartyScene}	&12	&51.0801 &13.6093 &13.5525 &11.9879	&\textbf{10.7588} &0.9984 &0.1160 &0.1481 &\textbf{0.1099} &0.1380 \\
			\cline{3-12} &24 &40.9438 &13.3090 &13.1844 &11.9118 &\textbf{10.3023} &0.9785 &0.1339 &0.1712 &\textbf{0.1054} &0.1504 \\
			\cline{3-12} &36 &33.7484 &13.3780 &13.1568 &11.7889 &\textbf{10.3876}	&0.8724 &0.1466 &0.2068 &\textbf{0.1279} &0.1498 \\  \hline
			\cline{3-12}\multirow{3}{*}{BasketballDrillText} &12 &51.3380 &11.4313 &11.9658 &\textbf{11.2467} &11.2885 &0.9965 &0.3335 &0.3466 &\textbf{0.3023} &0.3441  \\
			\cline{3-12} &24 &42.9425 &12.0454 &12.3082 &\textbf{9.9141} &10.4729 &0.9658 &0.3936 &0.4042 &\textbf{0.2728} &0.3514 \\
			\cline{3-12} &36 &36.3617 &11.6326 &11.9429 &\textbf{10.2188} &11.5484 &0.8780 &0.4068 &0.4683 &\textbf{0.3670} &0.4133 \\  \hline
			\cline{3-12}\multirow{3}{*}{Johnny}	&12	&51.8393 &14.0898 &14.3128 &13.2773	&\textbf{8.6507}	&0.9944 &0.4925 &0.5124 &0.4359	&\textbf{0.4259} \\
			\cline{3-12} &24 &46.6292 &14.2316 &13.6275 &12.9926 &\textbf{7.7676} &0.9725 &0.4706 &0.4471 &\textbf{0.3682} &0.4208 \\
			\cline{3-12} &36 &41.9026 &14.3882 &14.9000 &13.5148 &\textbf{9.1673} &0.9428 &0.5706 &0.4749 &\textbf{0.3905} &0.5031 \\ \hline
			\cline{3-12}\multirow{3}{*}{KristenAndSara}	&12 &51.8475 &15.7376 &15.0158 &14.0696 &\textbf{6.6476} &0.9944 &0.4474 &0.4690 &0.4022 &\textbf{0.3665} \\
			\cline{3-12} &24 &46.4661 &15.7631 &15.3568 &14.7831 &\textbf{6.8164}	&0.9765 &0.4253 &0.4882 &\textbf{0.3939} &0.4072 \\
			\cline{3-12} &36 &41.4027 &15.1599 &15.7244 &13.9391 &\textbf{9.8735}	&0.9517 &0.4779 &0.4899 &\textbf{0.4221}	&0.4959 \\  \hline
			\cline{3-12}\multirow{3}{*}{BasketballDrive} &12 &51.0072 &14.8807 &14.4820 &14.7322	&\textbf{10.9528} &0.9951 &0.5278 &0.5424 &\textbf{0.3372} &0.4328 \\
			\cline{3-12} &24 &44.0864 &13.2471 &14.9034 &14.5187 &\textbf{11.1040} &0.9419 &0.4941 &0.5185 &\textbf{0.3726} &0.4826 \\
			\cline{3-12} &36 &40.1487 &14.8385 &14.6580 &14.3899 &\textbf{11.3314} &0.8975 &0.5785 &0.5943 &\textbf{0.4335} &0.5282 \\  \hline
			\cline{3-12}\multirow{3}{*}{BQTerrace}	&12	&51.5661 &13.1145 &13.9256 &12.4986 &\textbf{8.6248} &0.9975 &0.2110 &0.2696 &\textbf{0.1761} &0.2170 \\
			\cline{3-12} &24 &42.7011 &15.5179 &16.8617 &14.2621 &\textbf{8.4283}	&0.9735 &0.2548 &0.2863 &\textbf{0.1904} & 0.2510 \\
			\cline{3-12} &36 &37.0627 &16.5752 &16.7083 &15.9872 &\textbf{8.6667} &0.8937 &0.3156 &0.2922 &\textbf{0.2348} &0.2651 \\  \hline
			\cline{3-12}\multirow{3}{*}{Traffic}	&12	&51.1566 &12.5901 &11.6006 &11.1289	&\textbf{9.9424}	&0.9960 &0.3295 &0.3392 &\textbf{0.2888}	&0.2919 \\
			\cline{3-12} &24 &43.0685 &11.5806 &11.1236 &11.0540 &\textbf{8.9662} &0.9774 &0.3464 &0.3574 &\textbf{0.2725} &0.3204 \\
			\cline{3-12} &36 &37.4918 &13.9138 &13.6104 &12.9100 &\textbf{9.3449} &0.9193 &0.3573 &0.3852 &\textbf{0.3075} &0.3286 \\  \hline
			\cline{3-12}\multirow{3}{*}{PeopleOnStreet}	&12	&51.3575 &13.2091 &12.6736 &12.4602 &\textbf{9.1224} &0.9965 &0.2911 &0.3119 &0.2442 &\textbf{0.2397} \\
			\cline{3-12} &24 &44.2459 &13.8099 &13.3416 &12.9737 &\textbf{9.2265} &0.9774 &0.3171 &0.3422 &\textbf{0.2437} &0.2455 \\
			\cline{3-12} &36 &39.4948 &13.7879 &12.6562 &11.8591 &\textbf{9.3852} &0.9137 &0.3591 &0.3572 &0.2943 &\textbf{0.2605} \\  \hline
	\end{tabular}}
	
	\footnotesize{PSNR: Peak Signal-to-Noise Ratio, QP: Quantization Parameter, RSVE: Robust Selective Video Encryption}
\end{table*}
\subsection{Edge Detection Analysis}
\begin{figure*}[!ht]
	\centering
	\includegraphics[width=0.7\textwidth]{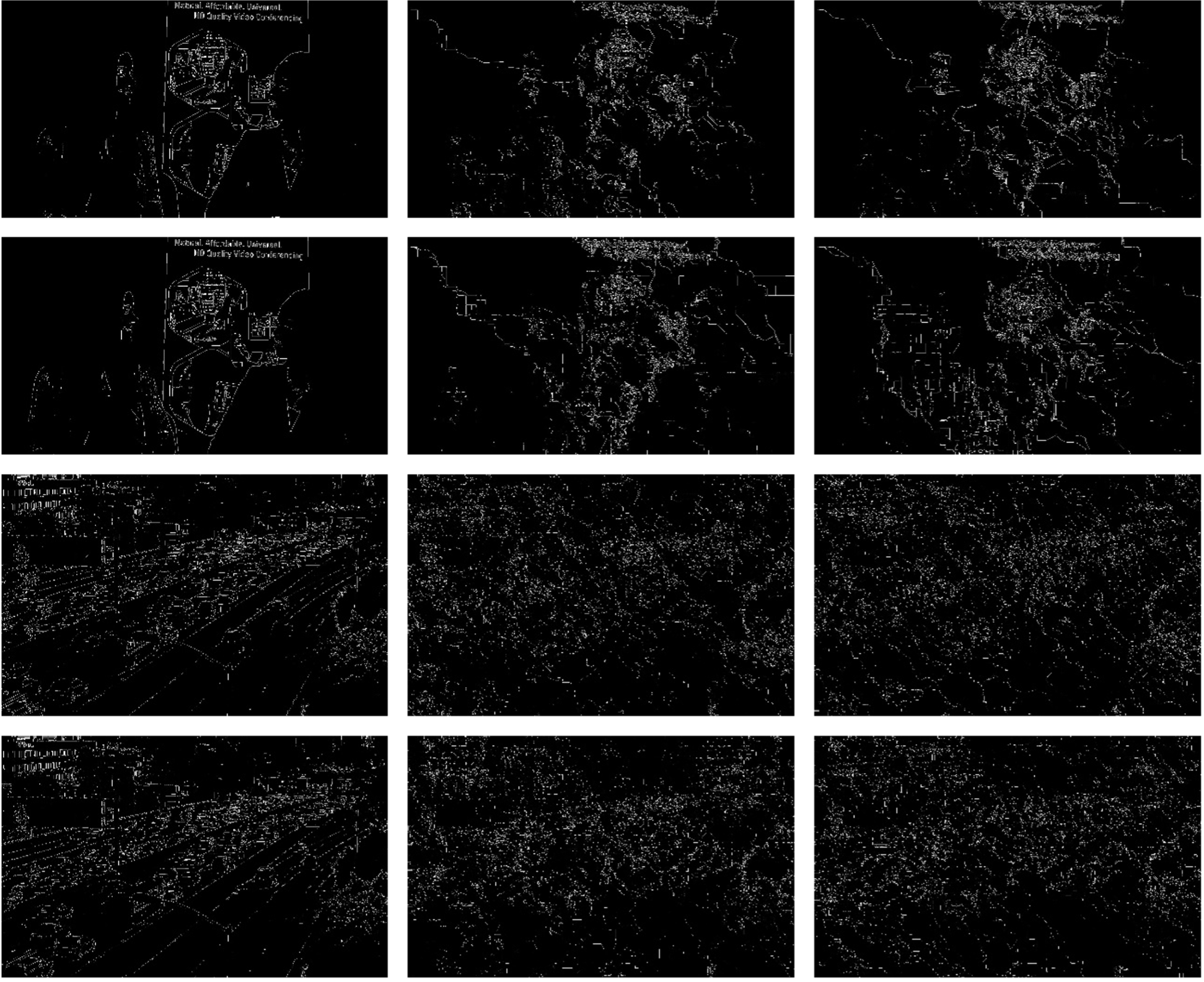}
\caption{Comparison of the edge images of KristenAndSara and the traffic sequence with the RSVE scheme. The first and second rows are the KristenAndSara video sequence (\#1 and \# 30 frame), and the third and fourth rows are the Traffic video sequence (\#1 and \# 30 frame). The first to third columns are the original frame, the RSVE scheme, and the proposed scheme frame, respectively. }
	\label{fig:edge}
\end{figure*}
To avoid revealing important information from the edge of the image, the encrypted edge frame should have a strong distortion effect to resist edge detection attacks. It is evident from Fig. \ref{fig:edge} that, compared with the original edge frame, the encrypted edge frame has a strong degree of distortion; hence, the content cannot be tracked; this shows the significance of the proposed scheme in effectively hiding edge information visually. 
Furthermore, Table \ref{tab:4} displays the evaluation of the edge difference ratio (EDR). All the results are averaged over 100 frames of EDR. The EDR indicators for the encrypted edge frames are close to 1. In comparison, to the RSVE scheme, the proposed scheme shows slightly better performance due to the consideration of coefficient scrambling conditions.
\begin{table}[!htbp]
	\caption{Edge difference ratio results}\label{tab:4}
	\centering
	\resizebox{1.0\columnwidth}{!}{
		\tiny
		\begin{tabular}{|l|c|c|c|c|c|}
			\hline 
			\multirow{2}{*}{\textbf{Video Sequence}} & \multicolumn{5}{c|}{\textbf{EDR}} \\ 
			\cline{2-6}&\textbf{Original} &\textbf{Jovanović \cite{c24}} &\textbf{Sallam \cite{c23-8119905}} &\textbf{RSVE \cite{c13-9903590}}&\textbf{Proposed} \\ \hline
			Mobile	&0.0569 &0.9219 &0.9172 &	\textbf{0.9339} &0.9296 	 \\  \hline 
			Foreman	&0.0853 &0.8833 &0.8993 &	\textbf{0.9345}	&0.9271 	 \\  \hline 
			BlowingBubbles	&0.1037 &0.8846 &0.8917 &0.9190	&\textbf{0.9241} 	 \\  \hline 
			RaceHorses	&0.1091 &0.9241 &0.9234 &0.9378	&\textbf{0.9415} 	 \\  \hline 
			PartyScene	&0.0719 &0.8923 &0.8933 &	\textbf{0.9243} &0.9216 	 \\  \hline 
			BasketballDrillText	&0.1050 &0.8256 &0.8347 &0.8885	&\textbf{0.8921} 	 \\  \hline 
			Johnny	&0.1291 &0.8380 &0.8441 &0.9034 &	\textbf{0.9121} 	 \\  \hline 
			KristenAndSara	&0.0818 &0.8043 &0.7995 &0.8711 &	\textbf{0.8807} 	 \\  \hline 
			BasketballDrive	&0.1206 &0.8980 &0.9001 &	\textbf{0.9454}	&0.9447 	 \\  \hline 
			BQTerrace	&0.0856 &0.8931 &0.8969 &	\textbf{0.9171}	&0.9146 	 \\  \hline 
			Traffic	&0.1361 &0.8896 &0.8866 &0.9228	&\textbf{0.9256} 	 \\  \hline 
			PeopleOnStreet	&0.1437 &0.9011 &0.9003 &0.9311 &	\textbf{0.9349} 	 \\  \hline 
	\end{tabular}}
\end{table}

\subsection{Chosen-Plaintext Attack Analysis}
The attacker creates some plaintext to generate the corresponding ciphertext and tries various means to obtain the correct key through ciphertexts. This type of attack is called a "chosen plaintext attack". $HD_{avg}$ is computed to calculate the difference between two sets of random sequences and to evaluate the resistance to chosen plaintext attacks, as shown in Table \ref{tab:pta}. The $HD_{avg}$ obtained for each group of videos is close to 0.5, indicating that there are large differences between the random sequences generated by each group of videos. The proposed scheme uses the SHA-384 hash value of the Slice Header to update the key and initial vector, so each slice has a different random sequence, indicating that the proposed scheme can effectively resist chosen plaintext attacks. 
\begin{table}[!htbp]
	\caption{Average hamming distance results for all two different videos with the same resolution}\label{tab:pta}
	\centering
	\resizebox{0.99\columnwidth}{!}{
			\tiny
			\begin{tabular}{|l|l|c|c|c|}
				\hline 
				\multirow{2}{*}{\textbf{Video Sequence 1}} & \multirow{2}{*}{\textbf{Video Sequence 2}} & \multirow{2}{*}{\textbf{Resolution}} &   \multicolumn{2}{c|}{\textbf{$HD_{avg}$}} \\
				\cline{4-5}	&&& \textbf{RSVE \cite{c13-9903590}} & \textbf{Proposed} \\ \hline 
				Mobile	&Foreman &352$\times$288&0.4968 &\textbf{0.4951} \\ \hline
				BlowingBubbles	&RaceHorses &416$\times$240&\textbf{0.4940} &	0.4952 \\ \hline
				PartyScene	&BasketballDrillText &832$\times$480&\textbf{0.4926} &	0.4939 \\ \hline
				Johnny&	KristenAndSara	&1280$\times$720&0.4991 &\textbf{0.4882} \\ \hline
				BasketballDrive	&BQTerrace &1920$\times$1080&0.4981 &	\textbf{0.4971} \\ \hline
				Traffic&	PeopleOnStreet &2560$\times$1600&0.4931 &\textbf{0.4906} \\ \hline
	\end{tabular}}
	
\end{table}
\subsection{Replacement Attack Analysis}
Attackers may try to modify the encrypted content in an attempt to find some clues about the encrypted video, which is called a "substitution attack". The experiment involves four substitution attacks on the PartyScene video for testing, namely:
\begin{itemize}
\item Replacing the encrypted Luma IPM with the Luma IPM as the first of the three most likely candidate prediction modes for the neighboring PU.
\item Replacing the encrypted QTC symbol with the QTC symbol =1
\item Replacing the encrypted MVD symbol with the MVD symbol =1
\item Replacing the encrypted MVD suffix with an MVD suffix =0.
\end{itemize}
\begin{figure*}[!ht]
	\centering
	\includegraphics[width=0.7\textwidth]{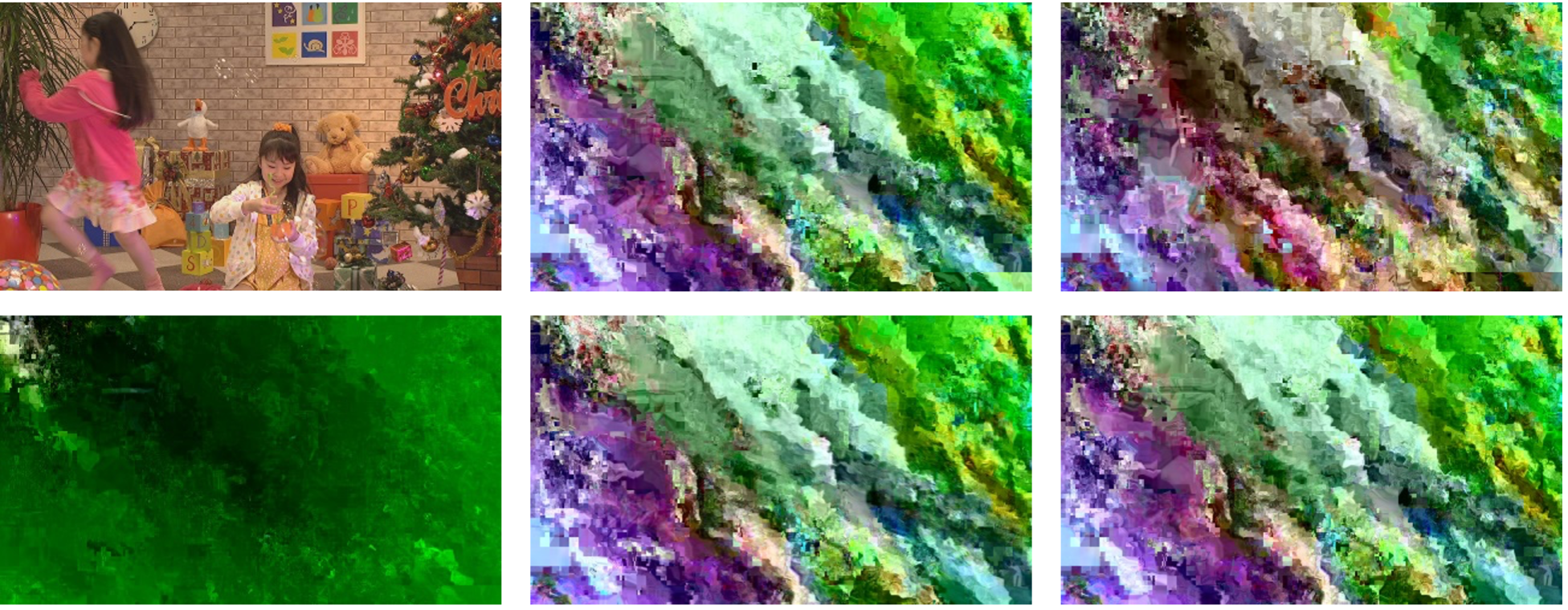}
\caption{Replacement attack analysis of the PartyScene video sequence (\# 30 frame). Original frame (upper left), encrypted frame (upper middle) and decrypted frames of the four substitution attacks. }
	\label{fig:raa}
\end{figure*}
From Fig. \ref{fig:raa}, it can be observed that the results of trying to use four alternative attacks still cannot reveal some information related to the original frame. 
\begin{table}[!htbp]
	\caption{PSNR and SSIM results of PartyScene video under different attack conditions} \label{tab:raa}
	\centering
	\resizebox{0.8\columnwidth}{!}{
			\tiny
			\begin{tabular}{|l|l|c|}
				\hline 
				\textbf{Video Status} & \textbf{PSNR} & \textbf{SSIM}\\ \hline 
				Original &34.1943&	0.9457 \\ \hline
				Encrypt	&10.8279	&0.1489 \\ \hline
				Replacement Attack 1&	11.2205	&0.1497 \\ \hline
				Replacement Attack 2&	10.0635	&0.2129 \\ \hline
				Replacement Attack 3&	10.7949	&0.1464 \\ \hline
				Replacement Attack 4&	10.8971	&0.1545 \\ \hline
	\end{tabular}}
\end{table}
Table \ref{tab:raa} portrays the results when the PSNR and SSIM indicators are used, considering that each frame is the result of the 30th frame of the PartyScene video. The decrypted frames of the four substitution attacks have smaller PSNR and SSIM values than the original frames do, and the PSNR and SSIM values are similar to each other compared with those of the encrypted frames. Thus, the proposed scheme can effectively resist substitution attacks. 
\subsection{Information Entropy Analysis}
Information entropy, also known as Shannon entropy, is used to measure the uncertainty of an event. The greater the information entropy, the more chaotic it is, and the probability of each event being average. Table \ref{tab:IEA} shows the information entropy index analysis results for various videos of 100 frames. The results indicate that the proposed scheme maintains a significant level of randomness in pixel distribution. However, the RSVE scheme demonstrates superior performance in terms of information entropy indices. 
\begin{table}[!htbp]
	\caption{Information entropy results}\label{tab:IEA}
	\centering
	\resizebox{1.0\columnwidth}{!}{
		\tiny
		\begin{tabular}{|l|c|c|c|c|c|}
			\hline 
			\multirow{2}{*}{\textbf{Video Sequence}} & \multicolumn{5}{c|}{\textbf{Information Entropy}} \\ 
			\cline{2-6}&\textbf{Original} &\textbf{Jovanović \cite{c24}} &\textbf{Sallam \cite{c23-8119905}} &\textbf{RSVE \cite{c13-9903590}} &\textbf{Proposed} \\ \hline
			Mobile	&7.6477 &7.8328 &7.8356 &\textbf{7.9159} &7.6351  \\  \hline 
			Foreman	&7.2290 &7.4967 &7.5597 &\textbf{7.6519}	&7.1786  \\  \hline 
			BlowingBubbles	&7.3596 &7.5867 &7.4515 &\textbf{7.6711}	&7.6489  \\  \hline 
			RaceHorses	&6.9685 &7.6297 &7.6359 &7.7302	&\textbf{7.7367}  \\  \hline 
			PartyScene	&7.4359 &7.8441 &7.8611 &\textbf{7.8621}	&7.8091  \\  \hline 
			BasketballDrillText	&6.9401 &7.6065 &7.8590 &7.5805	&\textbf{7.6999}  \\  \hline 
			Johnny	&6.8321 &6.9636 &6.8960 & \textbf{7.0704} &6.9851  \\  \hline 
			KristenAndSara	&7.1588 &7.1632 &7.1760 & \textbf{7.4700}	&7.3915  \\  \hline 
			BasketballDrive	&7.0192 &7.4958 &7.3175 & \textbf{7.6533}	&7.3984  \\  \hline 
			BQTerrace	&7.2192 &7.7355 &7.7098 &\textbf{7.8445}	&7.7709  \\  \hline 
			Traffic	&7.3522 &7.8077 &7.8345 &\textbf{7.9027}	&7.7848  \\  \hline 
			PeopleOnStreet	&7.0169 &7.8669 &7.8392 &\textbf{7.8877} &7.8615  \\  \hline 
	\end{tabular}}
\end{table}
\subsection{Bit Rate Overhead Analysis}
The values for the \textit{bit rate overhead}, an important parameter for effective video encryption, are shown in Table \ref{tab:5}. 
When coefficient scrambling is considered, the proposed scheme outperforms the RSVE scheme by an average of 45.13\% in terms of bit rate overhead. This suggests that the bit rate overhead can be significantly improved compared with that of the RSVE scheme.
\begin{table}[!htbp]
	\caption{Bit rate overhead results}\label{tab:5}
	\centering
	\resizebox{0.9\columnwidth}{!}{
		\tiny
			\begin{tabular}{|l|c|c|c|}
				\hline 
				\multirow{2}{*}{\textbf{Video Sequence}} & \multicolumn{2}{c|}{\textbf{Bit Rate Overhead}} &\textbf{Improvement} \\ 
				\cline{2-3} &\textbf{RSVE \cite{c13-9903590}} &\textbf{Proposed}& \textbf{Rate (\%)}\\ \hline
				Mobile	 &0.0566 &\textbf{0.0187} &66.96 \\ \hline
				Foreman	 &0.0714 &\textbf{0.0426} &40.34  \\ \hline
				BlowingBubbles	&0.0666 &\textbf{0.0350} &47.45  \\ \hline
				RaceHorses	 &0.0661 & \textbf{0.0286} &56.73  \\ \hline
				PartyScene	 &0.0706 &\textbf{0.0352} &50.14  \\ \hline
				BasketballDrillText	 &0.0799 & \textbf{0.0481} &39.80  \\ \hline
				Johnny	 &0.0764 &\textbf{0.0582} &23.82  \\ \hline
				KristenAndSara	 &0.0866 &\textbf{0.0590} &31.87  \\ \hline
				BasketballDrive	 &0.0861 &\textbf{0.0424} &50.75  \\ \hline
				BQTerrace	 &0.0566 & \textbf{0.0323}  &42.93  \\ \hline
				Traffic	&0.0697 & \textbf{0.0416} &40.32  \\ \hline
				PeopleOnStreet	 &0.0767 &\textbf{0.0380} &50.46   \\ \hline
	\end{tabular}}
	
\end{table}
\subsection{Key Space Analysis}
The attacker tries every possible key to crack the key, such attack is called a "brute force attack". To resist brute force attacks, the key space must be very large. The proposed scheme uses the AES encryption algorithm to create a random sequence. With a key size of 256 bits, there are $2^{256}$ possible combinations required to crack the key, indicating that the proposed scheme has an adequate ability to withstand attacks.
\subsection{Encryption Time Overhead Analysis}
The proposed algorithm demonstrates improved performance for low-resolution video sequences, making it suitable for mobile devices. Table \ref{tab:6} indicates that the proposed scheme performs comparably to the RSVE scheme across 12 videos. However, when considering the average encryption time overhead for all the videos, the RSVE scheme has an overhead of approximately 17.05\%, while the proposed scheme has only about 1.91\%. This finding indicates that the proposed scheme is highly efficient, offering better performance in terms of encryption time overhead.  
\begin{table}[!htbp]
	\caption{Encryption time overhead results}\label{tab:6}
	\centering
	\resizebox{0.7\columnwidth}{!}{
		\tiny
			\begin{tabular}{|l|c|c|}
				\hline 
				\multirow{2}{*}{\textbf{Video Sequence}}& \multicolumn{2}{c|}{\textbf{Overhead}} \\ 
				\cline{2-3} &\textbf{RSVE \cite{c13-9903590}} &\textbf{Proposed} \\ \hline
				Mobile	 &0.0207 &\textbf{0.008}   \\ \hline
				Foreman	 &0.0211 & \textbf{0.008}    \\ \hline
				BlowingBubbles	 &0.0127 &\textbf{0.010}    \\ \hline
				RaceHorses	&0.0227 & \textbf{0.016}   \\ \hline
				PartyScene	 &0.0174 &\textbf{0.009}   \\ \hline
				BasketballDrillText	 &0.0186 &\textbf{0.011}   \\ \hline
				Johnny	 &\textbf{0.0043} & 0.016  \\ \hline
				KristenAndSara	 &\textbf{0.0199}& 0.027 \\ \hline
				BasketballDrive	 &\textbf{0.0227} &0.024  \\ \hline
				BQTerrace &\textbf{0.0166} &0.034   \\ \hline
				Traffic	 &\textbf{0.0065} & 0.036  \\ \hline
				PeopleOnStreet	&\textbf{0.0214} &0.030   \\ \hline
	\end{tabular}}
	
\end{table}
\subsection{NPCR and UACI Analysis}
Table \ref{tab:NP} and Table \ref{tab:UA} present a comparison of the NPCR (Number of Pixels Change Rate) and UACI (Unified Average Changing Intensity) results for various video sequences, alongside findings from state-of-the-art works, including Wallendael et al. \cite{c4-6486782}, Sallam et al. \cite{c23-8119905}, and Boyadjis et al. \cite{c12-7370952}. The proposed scheme demonstrates optimal performance, outperforming the other approaches in terms of NPCR. Additionally, the UACI results are commendable, as they closely align with the expected value 0.334635.

\begin{table}[!htbp]
	\caption{Comparison NPCR results for various video sequences: proposed scheme v/s sota}\label{tab:NP}
	\centering
	\resizebox{1.0\columnwidth}{!}{
		\tiny
		\begin{tabular}{|l|c|c|c|c|c|}
			\hline 
			\textbf{Video} &\multicolumn{5}{c|}{\textbf{NPCR}} \\
			\cline{2-6}\textbf{Sequence}&\multirow{2}{*}{\textbf{Original}}&\textbf{Wallendael} &\textbf{Sallam}&\textbf{Boyadjis}&\multirow{2}{*}{\textbf{Proposed}} \\ 
			&&\textbf{\cite{c4-6486782}} &\textbf{\cite{c23-8119905}}&\textbf{\cite{c12-7370952}}&  \\ \hline 
			\cline{2-6}Mobile&	0.72235	&0.99559 &0.99562  &0.99569 & \textbf{0.99437} \\ \hline
			\cline{2-6}Foreman&0.70214	&0.99569 &0.99565 &0.99592 &\textbf{0.99837}  \\ \hline
			\cline{2-6}PartyScene	&0.72662 &0.99510 &0.99515 &0.99498 & \textbf{0.99665} \\ \hline
			\cline{2-6}Johnny	&0.65342 &0.99419 &0.99384 &0.99386 & \textbf{0.99801}  \\	 \hline
			\cline{2-6}BasketballDrive&0.69969 &0.99509 &0.99520 &0.99417 &\textbf{0.99612} \\ \hline
	\end{tabular}}
	
\end{table}
\begin{table}[!htbp]
	\caption{Comparison UACI results for various video sequences: proposed scheme v/s sota}\label{tab:UA}
	\centering
	\resizebox{1.0\columnwidth}{!}{
		\tiny
		\begin{tabular}{|l|c|c|c|c|c|}
			\hline 
			\textbf{Video} & \multicolumn{5}{c|}{\textbf{UACI}}\\ 
			\cline{2-6}\textbf{Sequence} &\multirow{2}{*}{\textbf{Original}}&\textbf{Wallendael} &\textbf{Sallam}&\textbf{Boyadjis}&\multirow{2}{*}{\textbf{Proposed}}\\ 
			&&\textbf{\cite{c4-6486782}} &\textbf{\cite{c23-8119905}}&\textbf{\cite{c12-7370952}}& \\ \hline 
			\cline{2-6}Mobile&0.01533 &0.28636 & 0.28667 &0.29201 &0.26266 \\ \hline
			\cline{2-6}Foreman&0.00982 &0.29351 &0.29340 &0.29130 &0.452675 \\ \hline
			\cline{2-6}PartyScene	&0.01509 &0.24876 &0.24931 &0.24495 & 0.32645 \\ \hline
			\cline{2-6}Johnny &0.00574 &0.26602 &0.26903 &0.26239 &0.397382 \\	 \hline
			\cline{2-6}BasketballDrive &0.00814 &0.23606 &0.23898 &0.21722 &0.264240 \\ \hline
	\end{tabular}}
	
\end{table}
\subsection{Comparison}
\subsubsection{Computational Complexity}
Table \ref{tab:complexity} analyzes the complexities of the proposed approach compared to state-of-the-art models. It considers $N$ as the number of syntax elements or transform coefficients processed, $M$ as the number of semantic features or elements extracted, $P$ as the parameters related to steganography payload size or hashing complexity, $K$ as the size parameter for SHA-256, and $H$ as the hierarchical levels in encryption. The computational complexities increase in the following order: Proposed Scheme $\equiv$ \cite{Comp4-10814092} $\equiv$ \cite{c5-8746776} $<$ \cite{c13-9903590} $<$ \cite{Comp2-doi:10.1142/S0218127424500135} $<$ \cite{Comp1-10556633} $<$ \cite{Comp3-10.1145/3698400}. 
\begin{table}[!htbp]
	\centering
	\caption{Comparison: computational complexity}	\label{tab:complexity}
	\resizebox{0.8\columnwidth}{!}{
		\tiny
			\begin{tabular}{|c|c|}
				\hline
				\textbf{Contributor} &  \textbf{Computational Complexity} \\ \hline
				Sheng et al. \cite{Comp4-10814092} & $\mathcal{O}(N)$  \\	\hline
				Chen et al. \cite{c13-9903590} &  $\mathcal{O}(N)+  \mathcal{O}(K)$\\ \hline
				Sheng et al. \cite{Comp1-10556633} & $\mathcal{O}(N) + \mathcal{O}(M)$ \\	\hline
				Peng et al. \cite{c5-8746776} & $\mathcal{O}(N)$ \\ \hline
				Sheng et al. \cite{Comp3-10.1145/3698400} & $\mathcal{O}(N) + \mathcal{O}(P) + \mathcal{O}(N \log N)$  \\	\hline
				Xing et al. \cite{Comp2-doi:10.1142/S0218127424500135} & $\mathcal{O}(N) + \mathcal{O}(H)$ \\	\hline
				Proposed Scheme & $\mathcal{O}(N)$ \\ \hline
	\end{tabular}}
\end{table}
\subsubsection{Comparative Analysis}
The proposed scheme is evaluated and compared with other related H.265 SEA studies using six key criteria (C1–C6) to ensure both security and practicality. C1 (\textit{Format Compliance}) checks whether the encrypted bitstream still adheres to the H.265/HEVC standard, allowing playback on standard decoders. It is marked as $\surd$ (Yes) if format compliance is preserved, or $\times$ (No) if the encryption breaks standard compliance. C2 (\textit{Distortion Level of the Encrypted Video}) measures how unintelligible the video becomes without decryption. It is marked as $\uparrow$ (High) for strong distortion with excellent confidentiality, $\leadsto$ (Medium) for partial distortion with moderate confidentiality, and $\downarrow$ (Low) for weak distortion with poor confidentiality. A higher distortion level is desirable, as it means the encrypted video is more secure and unintelligible to unauthorized viewers. If the distortion is low, the content may still be recognizable, leading to poor confidentiality. C3 (\textit{Execution Time Overhead}) assesses the extra computational cost introduced by encryption. It is marked as $\uparrow$ (High) for heavy time overhead (inefficient), $\leadsto$ (Medium) for moderate overhead (acceptable), and $\downarrow$ (Low) for minimal overhead (best performance). Lower overhead is preferable for real-time or large-scale applications, such as streaming, while high overhead can make the system impractical in real-world deployments. C4 (\textit{Bit Rate Overhead}) considers the impact on compression efficiency, as a minimal increase in bit rate preserves storage and transmission efficiency. It is marked as $\uparrow$ (High) for a significant increase in bit rate (poor efficiency), $\leadsto$ (Medium) for a moderate increase (tolerable), and $\downarrow$ (Low) for minimal increase (best case). Minimizing bit rate overhead ensures that encrypted videos remain efficient to store and transmit, whereas high overhead may strain bandwidth and storage. C5 (\textit{Robustness to Slice Loss}) evaluates whether decryption remains functional even if part of the video stream is lost, ensuring reliability in practical networks. It is marked as $\surd$ (Yes) if decryption remains functional despite slice loss, and $\times$ (No) if decryption fails or produces unusable results after slice loss. Finally, C6 (\textit{Resistance to Chosen-Plaintext Attack}) determines the scheme’s resilience against adversaries attempting to exploit known input-output pairs to break the encryption. It is marked as $\surd$ (Yes) if it is resistant to CPA (high security) and $\times$ (No) if it is vulnerable to CPA (weak security). Together, these six metrics comprehensively balance compatibility, confidentiality, efficiency, robustness, and security. 
The comparison results regarding format compatibility, visual distortion, chosen-plaintext attack, bit rate overhead, and robustness are presented in Table \ref{tab:7}. Most SEAs have characteristics that satisfy C1 and C3, but most of them are still insufficient in C2. Additionally, some SEAs that meet C2 still do not perform well in C4. The coefficient scrambling scheme in the RSVE method leads to a significant bit rate overhead. Therefore, the scheme proposed in this paper aims to satisfy C2 and all other five performance criteria. Consequently, the proposed scheme is competent for other H.265 SEA research. This reflects the significant reliability of the proposed approach for practical applications. 
\begin{table*}[!htbp]
	\caption{Comparative analysis: proposed scheme v/s state-of-the-art performance parameters}\label{tab:7}
	\centering
	\resizebox{0.9\textwidth}{!}{
		\tiny
		\begin{tabular}{|p{38pt}|p{85pt}|p{60pt}|c|c|c|c|c|c|c|}
			\hline 
			\textbf{Contributor} &\textbf{Encryption Method} &\textbf{Cipher Method} &\textbf{Year} &\textbf{C1} &\textbf{C2} &\textbf{C3} &\textbf{C4} &\textbf{C5} &\textbf{C6} \\ \hline  
			Lee et al. \cite{c2-9249233}	&NALU Header &	AES-ECB	&2020&	$\times$&$-$	&$\downarrow$	&$\downarrow$	&$\times$	&$\times$  \\ \hline 
			 Zhang et al. \cite{c8-9345211}	&MVD, QTC, Delta QP	&AES-CTR	&2020	&$\surd$	&$\downarrow$	&$\downarrow$	&$\downarrow$	&$\times$	&$\times$ \\ \hline 
			Zhang et al. \cite{c14-9390925}	&MVD, QTC, Delta QP	&Lorenz chaotic map	&2021	&$\surd$	&$\downarrow$	&$\leadsto$	&$\downarrow$	&$\times$	&$\times$ \\ \hline 
			Sallam et al. \cite{c6-Sallam2018}	&QTC, MVD, Delta QP, SAO sign, Reference frame index, Residual size	&RC6	&2018	&$\surd$	&$\leadsto$	&$\downarrow$	&$\downarrow$	&$\times$	&$\times$ \\ \hline 
			Wen et al. \cite{c9-Wen2023}	&QTC sign, Luma IPM	&AES-CTR	&2022	&$\surd$	&$\leadsto$	&$\downarrow$	&$\downarrow$	&$\times$	&$\times$ \\ \hline 
			Liu et al. \cite{c16}	&QTC, MVD, Delta QP sign, SAO sign, Merge index, Reference frame index, Partition mode	&Integer dynamic tent mapping	&2020	&$\surd$	&$\downarrow$	&$\downarrow$	&$\downarrow$	&$\times$	&$\times$ \\ \hline 
			Ye et al. \cite{c17}	&MVD, QTC, Delta QP	&L-ICMIC-CML chaotic map	&2021	&$\surd$	&$\downarrow$	&$\downarrow$	&$\downarrow$	&$\times$	&$\times$ \\ \hline 
			Peng et al. \cite{c5-8746776}	&Luma IPM, Chroma IPM, QTC, MVD, Merge index, MVP index, Reference frame index, SAO, Coefficient scrambling	&AES-CTR	&2020	&$\surd$	&$\uparrow$	&$\uparrow$	&$\uparrow$	&$\times$	&$\times$ \\ \hline 
			Chen et al. \cite{c13-9903590}	&QTC, MVD, Luma IPM, Coefficient scrambling	&Improved RC4	&2023	&$\surd$	&$\uparrow$	&$\downarrow$	&$\leadsto$	&$\surd$	&$\surd$ \\ \hline 
			Sheng et al. \cite{Comp1-10556633} & Chaos-Based Tunable Selective Encryption with Semantic Understanding &Logistic map-based chaos, Semantic content analysis& 2024 & $\surd$ & $\uparrow$ & $\leadsto$ & $\leadsto$ & $\times$ & $\surd$ \\ \hline 
			Xing et al. \cite{Comp2-doi:10.1142/S0218127424500135} &Hierarchical Multiscenario HEVC Video Encryption Scheme & Hierarchical selective encryption with chaotic sequences and symmetric cipher &2024 & $\surd$ & $\uparrow$ & $\leadsto$ &$\leadsto$ & $\surd$ & $\surd$  \\ \hline 
			Sheng et al. \cite{Comp3-10.1145/3698400} & Deep Hashing Network and Steganography Based Selective Encryption & Deep hashing, Steganographic embedding, Symmetric cipher & 2024 & $\times$ & $\uparrow$ & $\uparrow$ & $\uparrow$ & $\times$ & $\surd$ \\ \hline 
			Sheng et al. \cite{Comp4-10814092} & Content-Aware Tunable Selective Encryption using Sine-Modular Chaotification Model & Chaotic map (Sine-Modular) based coefficient scrambling &2025 & $\surd$ & $\uparrow$ & $\leadsto$ & $\downarrow$ & $\times$ & $\surd$ \\ \hline
			\textbf{Proposed Scheme} &QTC, MVD, Luma IPM, Coefficient scrambling &AES-CTR	&2025	&$\surd$	&$\uparrow$	&$\downarrow$	&$\downarrow$	&$\surd$	&$\surd$ \\ \hline 
	\end{tabular}}
	
	\footnotesize{$\uparrow$: High, $\downarrow$: Low, $\leadsto$: Moderate, $\surd$: Yes, $\times$: No, $-$: Unknown} 
\end{table*}
\section{Discussion and Limitations}\label{sec:dl}
The evaluation of the proposed chaotic map-based video encryption scheme across multiple quality, security, and efficiency metrics reveals that, while it demonstrates strong performance in several areas, it does not consistently outperform existing approaches in all scenarios. Each metric highlights specific strengths and trade-offs of the design. The higher PSNR after decryption indicates that the proposed scheme preserves visual quality more effectively than competing methods. This advantage arises from the selective and chaos-driven encryption design, which perturbs critical coefficients without introducing excessive reconstruction noise. As a result, the decrypted video closely matches the original content. However, despite the strong PSNR, the scheme underperforms in SSIM, suggesting that while overall pixel-level fidelity is high, subtle structural information and texture consistency are not fully preserved. This discrepancy may be attributed to chaotic modifications of coefficients, which can distort local structures and edges in ways that PSNR does not capture. SSIM's sensitivity to structural degradation explains why it presents lower scores. Additionally, the encrypted frames exhibit strong edge preservation, indicating that the chaotic scrambling does not oversmooth or blur significant features. This is particularly beneficial for applications where edge integrity is critical, such as surveillance. However, this may also explain the weaker SSIM, as excessive edge contrast can lead to reduced perceived structural similarity. Moreover, the scheme effectively resists chosen-plaintext attacks, owing to the sensitivity of chaotic maps to initial conditions and keys. Even minor differences in input lead to vastly different encrypted outputs, reinforcing its cryptographic strength. This characteristic shows that the design is well-suited against advanced adversarial models. However, the entropy values of the encrypted frames are lower, suggesting incomplete randomness in the ciphertext distribution. This limitation arises because the chaotic map, while non-linear, may introduce statistical bias in certain regions of the frame. The partial determinism in coefficient modification could reduce randomness, allowing for the presence of residual patterns. The scheme introduces minimal bit rate expansion compared to RSVE and similar methods. By carefully preserving compression-friendly structures in the H.265 stream, it avoids disrupting entropy coding efficiency, ensuring better bandwidth utilization for practical deployment. The computational burden of chaotic mapping is comparable to that of the RSVE scheme, balancing complexity with efficiency. While it does not outperform RSVE, it shows that the proposed scheme maintains real-time feasibility. This parity results from iterative chaotic computations, which, although secure, do not provide a speed advantage over lighter alternatives. Additionally, a high NPCR indicates strong sensitivity to plaintext changes, meaning the encryption effectively diffuses small input variations across the entire frame. The chaotic diffusion mechanism enhances pixel-level unpredictability, contributing to robustness against differential attacks. Although the NPCR is strong, the UACI results are comparable to existing methods, suggesting that while the encryption process introduces frequent pixel value changes, the magnitude of these changes is not sufficiently large across the image domain. This indicates that the diffusion effect is strong in terms of pixel position changes but limited in terms of altering pixel intensity amplitude. This can be explained by the selective encryption approach, in which only certain coefficients are modified, leading to limited amplitude variation. Overall, the inconsistency in performance across metrics reflects a core trade-off: the scheme prioritizes preserving compression efficiency (with low bit rate overhead) and ensuring strong cryptographic resistance (NPCR) at the expense of structural similarity (SSIM), entropy randomness, and UACI. Specifically, reliance on chaotic maps while beneficial for diffusion and attack resistance may create predictable statistical behavior in certain coefficient domains, reducing overall entropy. Similarly, selective encryption preserves format and bandwidth efficiency but inherently sacrifices structural integrity and complete intensity diffusion. 

The mixed performance of the proposed scheme highlights the need for targeted improvements to achieve more consistent results across all metrics. To address the underperformance in SSIM, information entropy, and UACI, future work could explore hybrid chaotic diffusion models that combine pixel and transform domain operations to better preserve structural similarity while enhancing randomness. Incorporating adaptive coefficient selection may balance distortion strength with compression efficiency, improving both SSIM and entropy. Furthermore, integrating multi-round chaotic transformations or entropy-boosting post-processing steps could enhance ciphertext unpredictability without significantly increasing bit rate or execution time. These refinements would strengthen the scheme’s robustness while maintaining its current advantages in PSNR, edge preservation, NPCR, and bit rate efficiency. 
\section{Conclusion and Future Work}\label{sec:con}
This paper presents a coefficient scrambling scheme based on a chaotic map by devising scrambling conditions for the coefficients wisely. It specifically focuses on encrypting three types of information: motion vector differences (MVD), quantization table coefficients (QTC), and Luma IPM (Luma Intra Prediction Modes). Additionally, the scheme synchronizes the generated random sequence with the slice and links it to the H.265 encoded stream. The experimental results demonstrate that the proposed scheme offers high security, compliance with formats, fast execution time, synchronous update with the slice, and resistance to common attacks. Compared with the RSVE scheme, it reduces the average bit rate overhead. Therefore, the proposed coefficient scrambling scheme presents a superior option for enhancing the security of H.265/HEVC video SEA. Future work will focus on developing a video encryption algorithm for critical areas via a 2D extended Schaffer function map and neural networks. Additionally, temporal action segmentation will be integrated to improve the accuracy, robustness, and adaptability of the model for video streams. 



\bibliographystyle{IEEEtran}
\bibliography{reference}

 \begin{IEEEbiography}[{\includegraphics[width=1in,height=1.25in,clip,keepaspectratio]{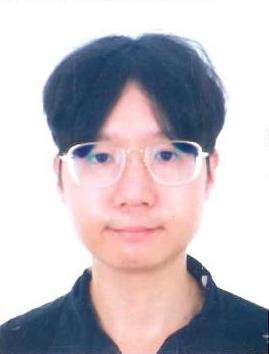}}]{Liang-Wei Li} has worked for his post graduation at Department of Computer Science and Engineering, National Sun Yat-sen University (NSYSU), Kaohsiung, Taiwan. His master's thesis research was based on the video encryption, especially the encryption of H.265 encoding. He had excellent academic performance in college received three certificates of excellence, and ranked in the top three overall. Currently, he is working as a Firmware Engineer.
\end{IEEEbiography}
\vspace{11pt}
\begin{IEEEbiography}[{\includegraphics[width=1in,height=1.25in,clip,keepaspectratio]{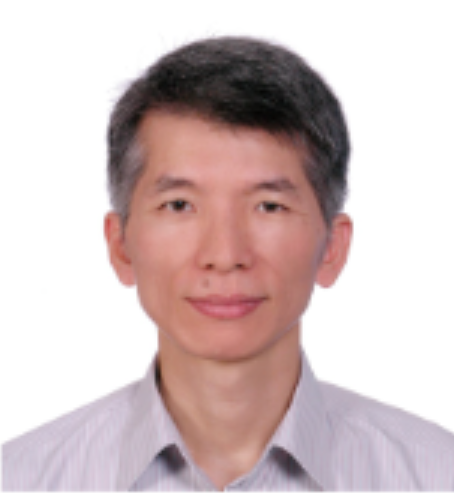}}]{Chung-Nan Lee} (Member, IEEE) received the B.S. and the M.S. degrees in electrical engineering from the National Cheng Kung University, Tainan, Taiwan, in 1980 and 1982, respectively, and the Ph.D. degree in electrical engineering from the University of Washington, Seattle, WA, USA, in 1992. Since 1992, he has been with the National Sun Yat-Sen University, Kaohsiung, Taiwan. He served as Chairman of the Department of Computer Science and Engineering, from 1999 to 2001. Dr. Lee was the President of the Taiwan Association of Cloud Computing, from 2015 to 2017, and the VP for TA of Asia-Pacific Signal and Information Processing Association, from 2019 to 2020. In 2016, he received an outstanding engineering Professor Award from the Chinese Institute of Engineers, Taiwan. 
	His research interests include multimedia over wireless networks, cloud computing, and IoT. 
\end{IEEEbiography}
\vspace{11pt}
\begin{IEEEbiography}[{\includegraphics[width=1in,height=1.25in,clip,keepaspectratio]{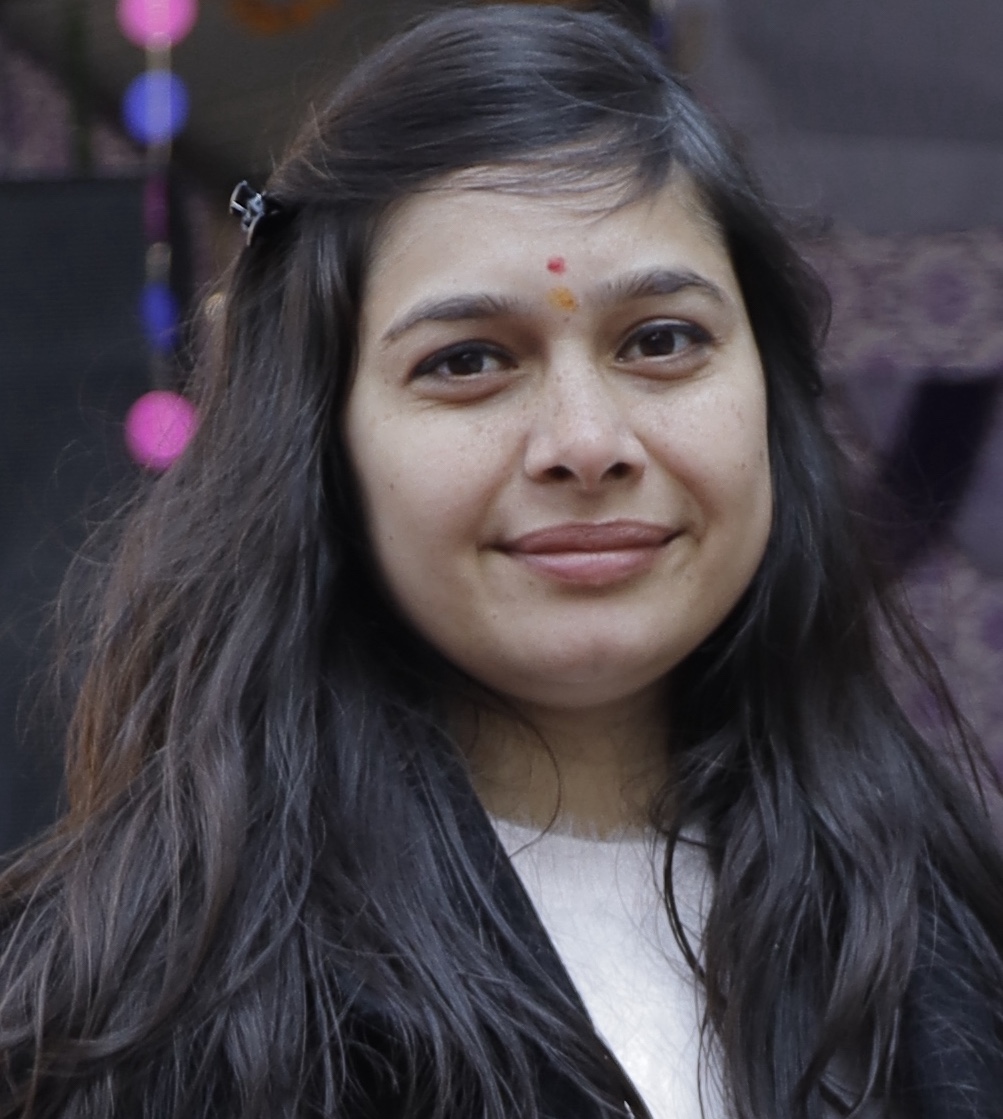}}]{Kishu Gupta} (Member, IEEE) received the Ph.D. degree in computer science from India in 2023. She is associated as a Post-Doctoral Researcher with National Sun Yat-sen University (NSYSU), Kaohsiung, Taiwan. She receiving a prestigious INSPIRE Fellowship sponsored by the DST, India. She has research findings published with top-notch venues, including IEEE TNNLS, IEEE TSMC, IEEE TASE, IEEE TITS, IEEE J-BHI, IEEE TCE, E-JNCA, Applied Soft Computing, Scientific Reports, Cluster Computing, and Procedia-Computer Science. Her major research interests include data security and privacy, evolutionary optimization, cloud computing, traﬃc management, federated learning, machine learning, neural networks, and quantum ML. Also, she was honored with a Gold Medal for her Academic Excellence (Ist Rank in M.Sc.), shortlisted for the Taiwan Comprehensive University System-Young Scholar (TCUS-YS) award, and received a Best Paper Award in RTIP2R-2024 Conference.
\end{IEEEbiography}
\vspace{11pt}
\begin{IEEEbiography}[{\includegraphics[width=1in,height=1.25in,clip,keepaspectratio]{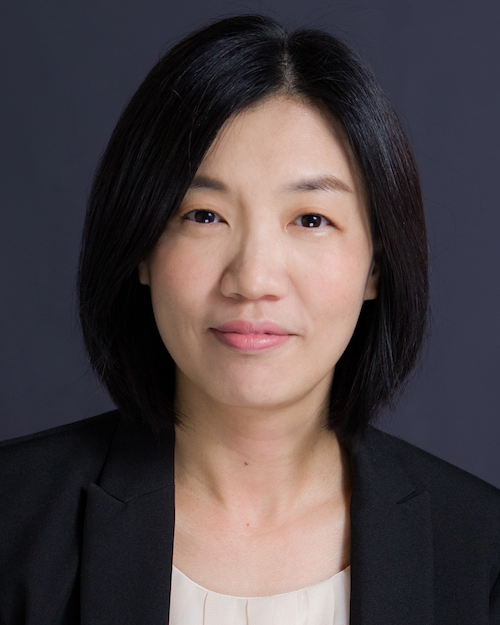}}]{Huei-Fang Yang} (Member, IEEE) received her Ph.D. in computer science from Texas A\&M University, College Station, TX, USA, in 2011. She is an associate professor in the Department of Computer Science and Engineering at National Sun Yat-sen University, Taiwan. Her research interests include computer vision and deep learning. In recent years, her work has focused on image retrieval and medical image analysis. 
	
\end{IEEEbiography}
\vspace{11pt}
\begin{IEEEbiography}[{\includegraphics[width=1in,height=1.25in,clip,keepaspectratio]{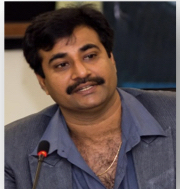}}]{Ashutosh Kumar Singh} (Senior Member, IEEE) is working as a Professor and Director of Indian Institute of Information Technology Bhopal, India. Also, he is working as Adjunct Professor in the VIZJA University, Warsaw, Poland. He received his Ph.D. in Electronics Engineering from Indian Institute of Technology, BHU, India and Post Doc from Department of Computer Science, University of Bristol, UK. He has research and teaching experience in various Universities of the India, UK, and Malaysia. His research area includes Design and Testing of Digital Circuits, Data Science, Cloud Computing, Machine Learning, Security. He has published more than 400 research papers in different journals and conferences of high repute. His research paper, published in the IEEE Transactions on Cloud Computing Journal, was honored with the 2022 Best Paper Award by the IEEE Computer Society Publications Board. 
\end{IEEEbiography}

\end{document}